\documentclass{article}

\usepackage[numbers]{natbib}
\usepackage{PRIMEarxiv}

\usepackage[utf8]{inputenc} 
\usepackage[T1]{fontenc}    
\usepackage{hyperref}       
\usepackage{url}            
\usepackage{booktabs}       
\usepackage{amsfonts}       
\usepackage{nicefrac}       
\usepackage{microtype}      
\usepackage{xcolor}         

\usepackage{natbib}

\usepackage{multirow}
\usepackage{subfigure}
\usepackage{booktabs} 
\usepackage{tabularx}
\usepackage{bm}
\usepackage{utfsym}

\usepackage{hyperref}

\usepackage{amsmath}
\usepackage{amssymb}
\usepackage{mathtools}
\usepackage{amsthm}
\usepackage{mathrsfs}

\theoremstyle{plain}
\newtheorem{theorem}{Theorem}[section]
\newtheorem{proposition}[theorem]{Proposition}
\newtheorem{lemma}[theorem]{Lemma}
\newtheorem{corollary}[theorem]{Corollary}
\theoremstyle{definition}

\newtheorem{assumption}[theorem]{Assumption}
\theoremstyle{remark}

\newtheorem{condition}[theorem]{Condition}

\usepackage{algorithmic}
\usepackage[linesnumbered,ruled,vlined]{algorithm2e}

\usepackage{graphicx}
\usepackage{xspace}

\newcommand{\bE}{\mathbb{E}}

\newcommand{\bH}{\mathbb{H}}

\usepackage{bbding}
\usepackage{extarrows}

\usepackage[shortlabels]{enumitem}
\newcommand{\zhiheng}[1]{{\textcolor{red}{#1}}}

\definecolor{myblue}{RGB}{0,0,255}

\title{Unveiling Environmental Sensitivity of Individual Gains in Influence Maximization
 
}

\author{
  Xinyan Su  $^{1,2}$\thanks{
\textbf{Equal contribution.}} \quad Zhiheng Zhang   $^{3}$$^*$,\quad Jiyan Qiu$^{1,2}$ \\
  $^{1}{\text{Computer Network Information Center, Chinese Academy of Sciences}}$ \\
  $^{2}{\text{University of Chinese Academy of Sciences}}$ \\
  \texttt{suxinyan@cnic.cn} \\ \\
  $^{3}{\text{Institute for Interdisciplinary Information Sciences, Tsinghua University}}$ \\
  \texttt{zhiheng-20@mails.tsinghua.edu.cn} \\
}

\begin{document}
\maketitle

\begin{abstract}
Influence Maximization (IM) is to identify the seed set to maximize information dissemination in a network. Elegant IM algorithms could naturally extend to cases where each node is equipped with a specific weight, reflecting individual gains to measure the node’s importance. In these prevailing literatures, they typically assume such individual gains remain constant throughout the cascade process and are solvable through explicit formulas based on the node’s characteristics and network topology. However, this assumption is not always feasible for two reasons: 1) \textit{Unobservability}: The individual gains of each node are primarily evaluated by the difference between the outputs in the activated and non-activated states. In practice, we can only observe one of these states, with the other remaining unobservable post-propagation. 2) \textit{Environmentally sensitivity}: In addition to the node’s inherent properties, individual gains are also sensitive to the activation status of surrounding nodes, which is dynamic during iteration even when the network topology remains static. To address these challenges, we extend the consideration of IM to a broader scenario with dynamic node individual gains, leveraging causality techniques. In our paper, we introduce a  Causal Influence Maximization (CauIM) framework and develop two algorithms, G-CauIM and A-CauIM, where the latter incorporates a novel acceleration technique. Theoretically, we establish the generalized lower bound of influence spread and provide robustness analysis. Empirically, in synthetic and real-world experiments, we demonstrate the effectiveness and robustness of our algorithms. 
\end{abstract}

\section{Introduction}

\label{intro}
Information propagation over networks has been booming in recent years. Due to the power of the "word-of-mouth" phenomenon, influence spread has been demonstrated as a necessity in various applications, such as viral marketing~\citep{chen2010scalable}, HIV prevention~\citep{wilder2018end} and recommendations~\citep{coro2021link}. The problem to select the seed set to maximize information spread is known as the \textbf{I}nfluence \textbf{M}aximization (IM)~\citep{kempe2003maximizing}.



Beyond optimizing the total number of infected nodes, current research has focused on investigating the individual gains of each node in real-world scenarios, referred to as weighted-IM~\citep{wang2016efficient, wang2017activity,han2021efficient}. Researchers endeavor to address the question: \textit{how limited resources can be utilized to maximize total gains?} This challenge manifests in various network scenarios, such as student networks and email networks, involving activities like awareness dissemination and product promotion. For example, when targeting users with varying purchasing power in product promotion, these users exhibit diverse purchasing behaviors  resulting in varying profits for the seller.  Here, regarding purchasing power as individual gains, the goal is to identify specific users for product advertising and and optimize the overall difference in profit gains pre- and post-product promotion dissemination.

Researchers usually assume that such purchasing power of each node remains observable and stable~\cite {kempe2015maximizingx} during the whole process. Such a weighted IM setting seems to be a natural extension of traditional IM and hence leads to relatively limited exploration~\cite{kempe2003maximizing}. However, in practice, this setting would be violated, and we summarize it as two fundamental properties as illustrated in Figure~\ref{intro-exp}
: \textit{(i) unobservability}. Accurately quantifying the actual purchasing power of each user is hindered by the limited observations of purchase occurrences (where only two outcomes are observable: activated or not, corresponding to purchase or non-purchase for each node; one represents the ``factual'', and the other ``counterfactual''), thereby complicating the determination of the actual increase in benefit gains that each user can deliver to sellers;  \textit{(ii) environmental sensitivity}. The expected purchasing power of each person is not only connected to the individuals themselves but also influenced by the attitudes of their social contacts. For example, as more people in the friend circle make purchases, individuals become increasingly susceptible to influence and are likely to make additional purchases. Such properties emphasize that individual gains in IM are environmentally dynamic and challenging to ascertain, posing challenges in their computation.

To tackle these challenges, we employ causality techniques to model the advertising problem, inspired by the concept of Individual Treatment Effect (ITE)~\citep{pohl2008modelling}. It measures the disparity between the factual world and the counterfactual ones, which goes beyond traditional observational studies, such as those based on network structure or direct assignment of other feature weights~\citep{kempe2003maximizing,gao2020viral}. By incorporating the process of inference, we transition the problem from observational study to the direct utilization of ground truth for measuring individual gains at the node level. Drawing from this, we propose the \textbf{Cau}sal \textbf{I}nfluence \textbf{M}aximization (CauIM) framework. Specifically, in the hypergraph modeling\footnote{The utilization of IM in \text{hypergraphs}~\citep{antelmi2021social, xie2022influence} introduces a higher-order structure that establishes connections between clusters,  effectively reflecting real-life relationships on general graph, and surrogates the traditional normal graph as a special case.}, we redefine the objective function of traditional IM by incorporating ITE as weights assigned to each node. It incorporates the internal covariate information and external environmental information for each node (Figure ~\ref{intro-exp}).


\begin{figure*}[t]
    \centering
    \includegraphics[width = 12cm]{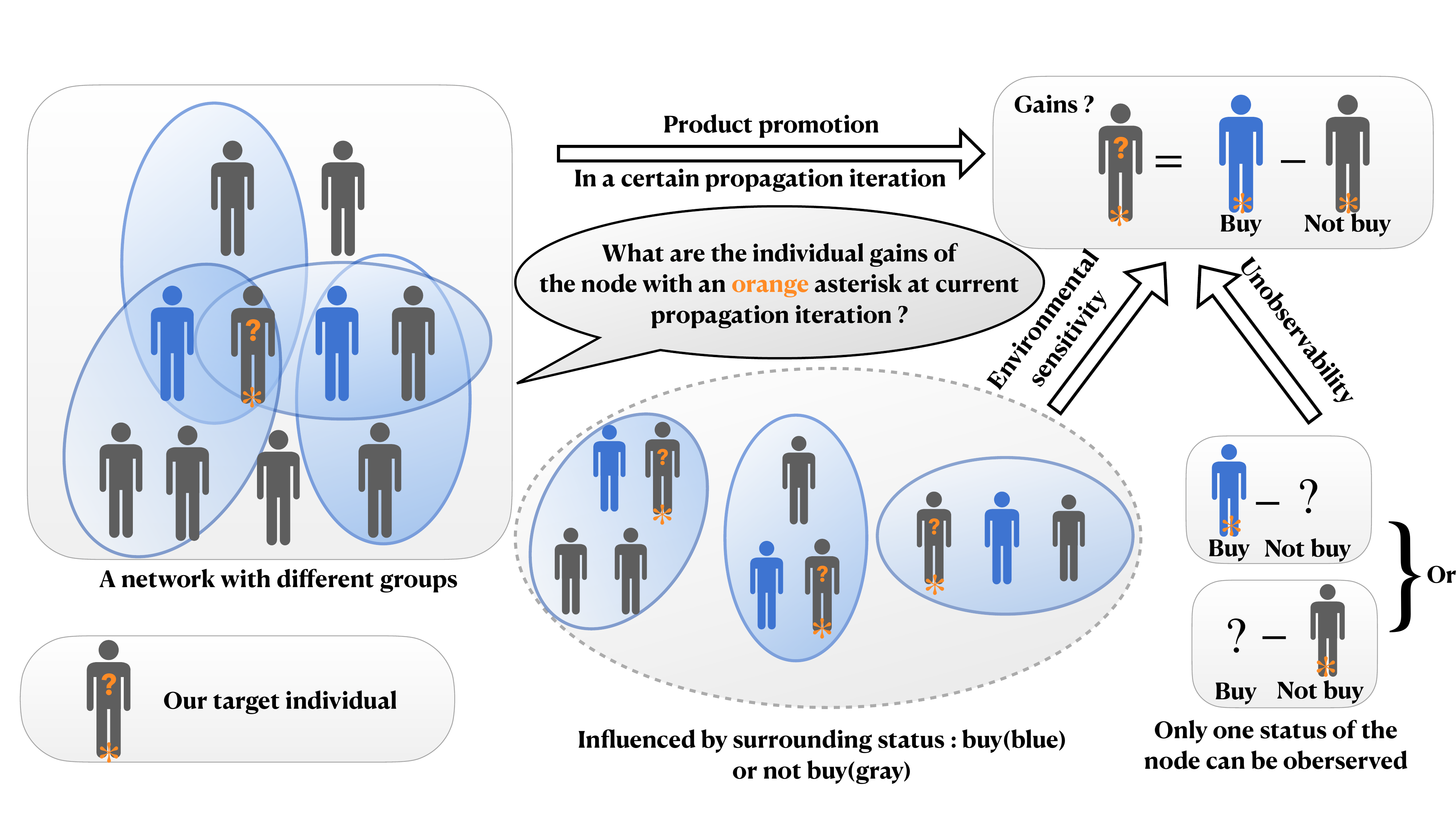}
    \caption{Illustration of individual gains during a certain propagation iteration in product promotion scenario, where we focus on the analysis of the starred (*) node. The current iteration is depicted in the leftmost figure, illustrating nodes in varying states within the network; blue nodes indicate activation for purchases, while grey nodes remain inactive, and $?$ denotes unkonwn status. The actual individual gains is the difference between the profit of the node in the activated state and that in the non-activated state. The first challenge is located on the \textit{unobservability}, i,e, we can only observe one status of these nodes, with the counterfactual scenario unknown. More seriously, the second challenge is \textit{environmental sensitivity}, which indicates the individual gains of the node are affected by the activation status of others.  } 
    \label{intro-exp}
\end{figure*}

Noteworthy, it is not merely a simple causality-plug-in interdisciplinary attempt since we should proactively challenge the Stable Unit Treatment Value Assumption (SUTVA) in causality ~\citep{angrist1996identification}, which permits the presence of interference between different nodes as illustrated above. Taking a step forward, even if there is pioneering exploration upon interference-based causality~\citep{ma2022learning,ma2022clear,leung2022causal}, distinguishing CauIM from such literature, we delve deeper into practical constraints: the original treatment policy can solely influence a ``limited seed set", and then the causality estimator should be considered under a dynamic propagation process with the challenge of heavy time-consuming and instability. Such process could be highly non-trivial since the environment information for each node is dynamic. In this
sense, we further design a fast and stable acceleration framework, replacing the inefficient greedy selection process with vectorized fast derivative operations.

In sum, 
our contributions are summarized as follows: \textbf{i)} We propose CauIM for networks that consider communities and environmental-sensitive node individual effects to achieve influence maximization and justify its practical scenarios. \textbf{ii)} We show the effectiveness (approximate optimal guarantee), robustness and acceleration feasibility of CauIM framework, especially when ITE is permitted to be not always positive with existing estimation bias. In correspondence, we provide the greedy-based implementation (G-CauIM) and further design Gradient Descent-based Accelerated CauIM~(A-CauIM). \textbf{iii)} We conduct experiments of {CauIM} on three real-world datasets and one synthetic dataset. It not only supports our theoretical claim upon effectiveness and robustness, but also demonstrates the practical efficiency improvement of A-CauIM upon traditional IM.

\section{Preliminary}\label{preliminary}

\paragraph{Notations and Basic Concepts}\label{Notation} We develop our model on an undirected hypergraph $\mathcal{G}(\mathscr{V}, \mathscr{H}, \mathbb{H})$, where $\mathscr{V} := \{v_1,v_2,...v_n\}$, $\mathscr{H} := \{h_1,h_2,...h_m\}$, representing the node set and the hyperedge set respectively. In this context, each undirected hyperedge signifies a social connection of the nodes inside it, and $\mathbb{H}$ denotes the hypergraph structure (adjacent matrix between hyperedge and node). For each node $v_i$, covariate $X_{v_i}$ denotes its node feature and $\mathcal{N}_{{v_i}}$ indicate the set of its neighborhood\footnote{Here $v_i$ and $v_j$ are neighborhoods indicate that they are within at least one hyperedge.}. To maintain clarity, we abuse the notation that $X_{v_i}=X_{i}$ and $N_{v_i}=N_{i}$, applying the same convention for symbols with the subscript $v_i$.  

Stepping forward, we would like to introduce several broad concepts in causal inference: the potential outcome (individual profits of one status) of each node $v_i$ is denoted as $Y_{i}(T_{i}=t; X_{i}, \bm{T}_{-{i}}, \bm{X}_{-{i}})$. Here $t=0,1$ refer to the case where node $v$ is activated in the diffusion process or not\footnote{Here $Y(\cdot)$ is a function of $\{T_{i}; X_{i}, \bm{T}_{-{i}}, \bm{X}_{-{i}}\}$. The first two items refer to the inherited information of node ${v_i}$, and the last two items refer to the environment information. We omit the information $\mathbb{H}$ in the mapping process since it remains stable in this paper. We defer the possibility of dynamic graphs in future research.}. Moreover, $\bm{T}_{-{i}}, \bm{X}_{-{i}}$ refer to the environment information, namely, $\bm{T}_{-{i}} := \{T_{1}, T_{2}, ...T_{{i-1}}, T_{{i+1}},...T_{n}\}$, $\bm{X}_{-{i}} := \{X_{1}, X_{2}, ...X_{{i-1}}, X_{{i+1}}, ...X_{n}\}$. The individual treatment effect (ITE) is defined as
\begin{equation}\begin{aligned}\tau_{ i} := Y_{ i}(T_{ i}=1; X_{ i}, \bm{T}_{-{ i}}, \bm{X}_{-{ i}}) \\- Y_{ i}(T_{i}=0; X_{i}, \bm{T}_{-{i}}, \bm{X}_{-{i}}). \label{causal_effect} 
    \end{aligned} \end{equation}
    Here ``treatment'' \{$T_{i}, \bm{T}_{-{i}}$\} indicate the activating status of $v_i$ and its surrounding nodes. ITE represents the difference of outputs between the activated and inactivated cases. As illustrated in the introduction, it could not be directly extracted from observations (\textit{property (i)}), and it also depends on the activation state of surrounding nodes (\textit{property (ii)}).

\label{formulation}
\paragraph{Problem Formulation} We consider the Susceptible-Infected Contact Process (SICP) diffusion model, which is well-known in research of hypergraph such as~\citet{xie2022influence}. Initially, we choose a  seed set $S_0$. Then in every iteration of propagation, each activated node $v$ selects one of its affiliated hyperedge randomly and subsequently affects its inactivated neighbors $u\in\mathcal{N}_v$ within this hyperedge with a specific probability. The total probability of $u$ activating $v$, including hyperedge choosing process in one step, is denoted as $p_{uv}$. This process is repeated (Figure~\ref{pic1} in Appendix) until no more nodes are activated. Such extension of the traditional IC model to the hypergraph has potential applicability across various scenarios~\citep{suo2018information}. Noteworthy, we also have the flexibility to select alternative propagation models, e.g., Linear Threshold~(LT)~\citep{granovetter1978threshold, goyal2011simpath}, to which our algorithms described in Section~\ref{methodology} are also applicable.

 Here, we define $ap(v_i;S_0)$ as probability of node $v_i$ getting infected in the entire propagation process with seed set $S_0$.  According to~\citet{wang2016efficient}, $ap(v_i;S_0) = \sum_{u \in S_0}p_r(u,v_i)$, where $p_r(u,v_i)$ denotes the probability of reachability from node $u$ to $v_i$ inclusive of all reachable paths. Finally, we identify the objective function as the expected total causal inference during the diffusion process: $\sigma(S) = \mathbb{E} \left[ \sum_{{v_i} \in \mathscr{V}} ap(v_i; S) \tau_{i}\right], S \subseteq \mathscr{V} $. Here $\tau_i$ is identified in Equation~\ref{causal_effect} and the expectation takes upon activation status $\{T_i, \bm{T}_{-i}\}$ in all possible propagation process. In conclusion, our {CauIM} can be formalized as follows:
\begin{equation}
    \begin{aligned}
    S^*=\operatorname{argmax}_{S}\{  \sigma(S) \},
s.t. |S|\leq K.
    \end{aligned}\label{final_object}
\end{equation}
Here $K$ is a constant. Notice that compared to the traditional definition, we are able to obtain such a concise and general expression at the cost of introducing the relatively difficult-to-calculate parameter $ap(\cdot)$. Due to its difficulty on computation, we will elaborate our new efficient computational process in the next section.

\begin{assumption}[Basic assumption] \label{basis_ass}
1) Bounded ITE: The ITE is bounded as $\max_{v_i \in\mathscr{V}} |\tau_i| \leq M$, where $M$ is a constant. 2) Consistency~\citep{cole2009consistency}: The potential outcomes $Y_{i}(T_{i}=t; X_{i}, \bm{T}_{-{i}}, \bm{X}_{-{i}})$ are deterministic in Equation~\ref{causal_effect} and are equal to the observational values of $Y$ for $t =0,1$ given fixed $\{ X_{i}, \bm{T}_{-{i}}, \bm{X}_{-{i}}\}$.

\end{assumption}

These standard assumptions for ITE refer to~\citet{pearl2018book}. Noteworthy, Assumption~\ref{basis_ass} has also been analyzed the interference effect, which is well-known as violating the {SUTVA~\citep{imbens2015causal}}, called as ``the potential outcomes for any unit do not vary with the treatment assigned to other units''. 
Such information is surrogated into the environment function $O_{i} = \text{ENV}(\bH, \bm{T}_{-{i}}, \bm{X}_{-{i}})$. This function follows~\cite{ma2022learning}, where the other covariate and treatment assignment is summarized as a low-dimensional vector. We formalize it into the Assumption~\ref{env_ass}~\citep{ma2022learning}.


\begin{assumption}[Environment assumption~\citep{ma2022learning}] \label{env_ass} For each node $v_i$, these two potential outcomes in Equation~\ref{causal_effect} are independent with each other given $\{T_{i},O_{i}\}$.

Assumption~\ref{env_ass} extends the {ignorability assumption~\citep{joffe2010selective}} to the graphical case, where authors claimed that the pair of potential outcomes is independent of the treatment assignment, given the covariates of each node. It just guarantees there is no unmeasured covariates in the graph, which is fairly broad and has been adopted by~\citet{ma2022learning}. Assumption~\ref{env_ass} essentially ensures that the two effects ($\{X_{i},T_{i}\}, \{\bm{X}_{-{i}},\bm{T}_{-{i}}\}$) together constitute a sufficient statistic for ITE. In other words, ITE could be legitimately estimated via observations under these assumptions and hence provides the potential to design the estimator \footnote{Under these two assumptions, it has been demonstrated that the expected potential outcome of $v_i$ could be computed by observational data~\citep{ma2022learning}, namely, $\bE (Y_{i}(T_{i}=t; X_{i}, \bm{T}_{-{i}}, \bm{X}_{-{i}}) )= \bE(Y_{i} \mid X_{i} = x_{i}, T_{i} = t, o_{i} ), t \in \{0,1\}$.}. 

\end{assumption}

    
       
       
       
       
       
       
       
       
       


\section{Methodology}\label{methodology}

In this section, we first identify three primary challenges when achieving our algorithms. Subsequently, we provide a detailed introduction to two algorithms within our CauIM framework. The first one is offline greedy-based implementation for causal influence maximization named G-CauIM (Algorithm~\ref{alg1: traditional_cauim}). We then improve the efficiency of G-CauIM by speeding up the diffusion and greedy selection process, and propose a Gradient Descent-based Accelerated CauIM, referred to as A-CauIM (Figure ~\ref{alg-gd}).  Finally, we theoretically demonstrate the algorithm's effectiveness. To begin with, we summarize three main challenges as follows.




\subsection{Three Challenges}
(i) Unmeasured individual effect(ITE) (Equation~\ref{causal_effect}): the inherent limitations in causal inference necessitate the recovery of the counterfactual to address the “missing data problem” in our objective function Equation~\ref{final_object}. Furthermore, it may vary across iterations due to its dependence on environmental information ${\bm{X}_{-{i}}, \bm{T}_{-{i}}}$. (ii) Approximate optimal guarantee: The traditional greedy-based IM might not guarantee sub-optimal properties due to the unknown individual effect as mentioned above and, therefore, requires re-analysis. (iii) Estimation bias: CauIM exhibits robustness against biases stemming from individual effect estimation and the sampling strategy.

\begin{algorithm}[ht]
  \caption{G-CauIM}
  \label{alg1: traditional_cauim} 
  \KwIn{$\mathcal{G}(\mathscr{V}, \mathscr{H}, \mathbb{H})$; seed number $K$; $X_{i}$,  initial treatment $T_{i}$  and $\bm{T}_{-{i}}$ of each node ${v_i}$; observational data $D=\{Y_{i}(t; \cdot)\}_{v_i \in \mathscr{V}}$, where $t = T_i$; {the bound} $M$ for ITE.}
  \KwOut{Deterministic seed set $S^*$ with $|S^*| = K$.}
 
 \SetKwProg{Fn}{Function}{:}{end}
    \Fn{$\widehat{\rm{ITE}}$($X_{i}$, $T_{i}$,$\bm{T}_{-{i}}$,$\mathcal{G}; \theta$)}{
        Compute the representation $Z_{i}$ of $X_{i}$ via  representation learning; 

        Compute the high-order interference representation $O_{i} := \text{ENV}(\mathbb{H}, \bm{T}_{-{i}}, \bm{Z} _{-{i}}; \theta)$ (Assumption~\ref{basis_ass} and Assumption~\ref{env_ass} );
        
        Concatenate $Z_{i},~Q_{i}$ and feed them into a Multi-Layer Perception (MLP): $\{\hat{Y}_{i}(1;\cdot), \hat{Y}_{i}(0;\cdot)\} \sim {\rm{MLP}_{}}([Z_{i} || O_{i}])$;
         
         Compute the ITE  $\hat{\tau}_{i}$ = $\hat{Y}_{i}(1; \cdot) - \hat{Y}_{i}(0; \cdot)$ for $v_i$;
         
        \Return{$\hat{\tau}_{i}$;}
        
    }

    \Fn{${\textbf{Main}}$}{
    (Initialization) $S^* = \emptyset$;  $\rm{Loss} = 0$;
    

 { (Training): Compute the cumulative loss by $D$: $\text{Loss} = \sum_{v_i \in \mathscr{V}, t=0,1}\big|\big(\hat{Y}_{i}(t;\cdot) - {Y}_{i}(t;\cdot)\big)\mathbb{I}(T_{i} = t)\big| $ via the above $\widehat{\rm{ITE}}(\cdot; \theta)$ function;
 
 Get the optimal $ \theta^{opt}:= \arg\min \{\text{Loss}: \forall v_i, \hat{\tau}_{i} \leq M\}$;
  
}
  
  \For{$|S^*|<K$}
  {
    Conduct propagation under current seed set $S$, generate ${\hat{\tau}}_{i}$ = \textit{$\widehat{\rm{ITE}}({X}_{i}$, $T_{{i}}$,$\bm{T}_{-{i}}$,$\mathcal{G}; \theta^{'})$} for $v_i \notin S$, where $T_{i}$ is changed to its current activated state( 0 or 1), and $\bm{T}_{-{i}}$ is changed based on other nodes' activated states, $\theta^{'}:= \theta^{opt}+\triangle_\theta, \triangle_\theta:=\min\{\|\theta_q\|: \exists \delta \leq \|\theta_q\|, \widehat{\rm{ITE}}(\cdot; \theta+\delta) \leq M \}$, repeat the process and get the mean;


   $v_0 = argmax_{v \notin S^*}\left\{\sigma \left(S^*\cup{\{v\}}\right)-\sigma(S^*)\right\}$; 
   
   $S^*=S^*\cup{\{v_0\}}$;
   
  }

    }
    \Return $S^*$.
  
\end{algorithm}
\noindent \textit{Address challenge (i)}: ITE estimation. Motivated by~\citet{ma2022learning, ma2021causal}, we recover the individual ITE from observational data using a neural network model represented by $\widehat{\rm{ITE}}(\cdot)$ function in both G-CauIM and A-CauIM. To handle dynamic characteristics of ITE, we incorporate a model parameter  adjustment strategy into Algorithm~\ref{alg1: traditional_cauim} of G-CauIM. Additionally, we employ an approximation strategy in A- CauIM, as depicted in the ITE Estimator section of Figure~\ref{alg-gd}.  Further details of this procedure are provided in the subsequent Algorithms section.

\noindent \textit{Address challenge (ii)}: To achieve approximate optimum, we inductively select the seed candidate via greedy search. We show this greedy strategy will also hold a weaker but analogous order of ($1-\frac{1}{e}$) approximate optimum level of the traditional IM in further theoretical parts: 
\begin{equation}
    \begin{aligned}
    v_0 = argmax_{v \notin S^*}\left\{\sigma \left(S^*\cup{\{v\}}\right)-\sigma(S^*)\right\}.  
    \end{aligned}\label{greedy_alg}
\end{equation}

 It can be seen as a more general result since the traditional IM serves as the special case with $\tau = 1$. However, the analysis of the above settings is harder since the submodularity might not always exist due to the potential negative ITE. Additionally, for practical issues, We refer readers to Appendix~\ref{related} for the Monte Carlo-based greedy {CauIM}.

\noindent \textit{Address challenge (iii)}: Our CauIM model demonstrates robustness against bias in estimating Individual Treatment Effects (ITE), which will be outlined in the theoretical discussion. This resilience is ensured under broad assumptions regarding the controlled probability $p_r(v,u)$ of reaching all nodes $u\in\mathscr{V}$ in the complete graph from the current propagation node $v$, which is readily feasible in real-world scenarios, as evidenced in the experimental findings. In a word, this property provides us with more flexibility and possibilities in adjusting the estimator parameters.



\subsection{Proposed Algorithms}
\paragraph{G-CauIM.}Our primary procedures are executed in function $Main$ (line {7}). We initially train the ITE estimation model offline (line {9}), represented by function $\widehat{\rm{ITE}}(\cdot)$ (line~{1}). In this model, for node $v_i$, we construct $Z_{i}$ via each covariate $X_{i}$ using representation techniques (line {2}). Additionally, we construct $O_{i}$ to denote the higher-order interference representation of node $v_i$  with its environments: $O_{i}=\text{ENV}(\mathbb{H},\bm{T}_{-{i}}, \bm{Z}_{-{i}})$ (line {3}). Here $\text{ENV}$ is a transmission function using Hypergraph Convolution module~\citep{bai2021hypergraph} with more details in Appendix~\ref{related}, and $\bm{Z}_{-{i}}:= \{Z_{1},...Z_{{i-1}}, Z_{{i+1}}, ...Z_{n}\}$. Finally, with the combination of representation $Z_{i}$ and  $O_{i}$, we obtain the estimation value $\hat{Y}_{i}(1), \hat{Y}_{i}(0)$ via $\text{MLP}$ model (line {4}). During training, we employ a balancing mechanism to ensure covariate balance between the treatment group $\{v_{i}: T_{i}=1\}$ and control group $\{v_j: T_{j}=0\}$, achieved by incorporating an additional penalty term to the representation vector.  Such technique is not unique, referring to ~\citet{yao2018representation, harshaw2024balancing}.  In addition, line {10} is to ensure the estimated ITE is bounded by $M$ (identified in assumption~\ref{basis_ass}) via controlling the estimator parameter $\theta$. Diffusion process and greedy selection take place in line {11-14}, where we employ the traditional greedy algorithm strategy and re-analyze the marginal benefit function as the incremental form of ITE. It is noteworthy that as the seed set expands, the activated states of each node change in the propagation, leading to diverse values of $\hat{\tau}_i$. To bound the varying $\hat{\tau}_i$, we  adjust the parameters $\theta^{opt}$ of the trained $\widehat{\rm{ITE}}(\cdot)$, as depicted in line 12.


\begin{figure}[t]

    \centering
    
    \includegraphics[width = 8cm]{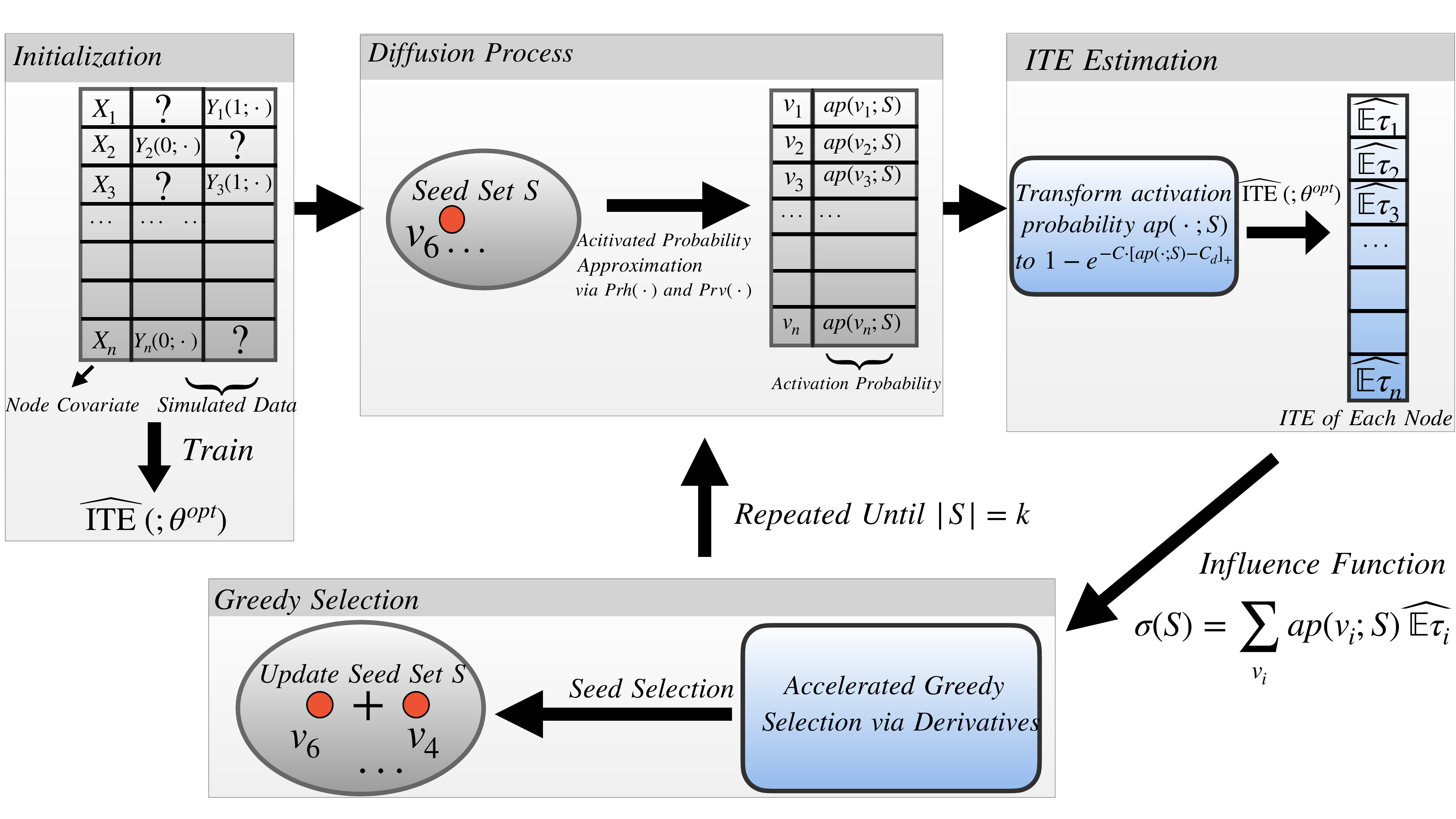}
    
    \caption{{A-CauIM. Compared with G-CauIM (Algorithm~\ref{alg1: traditional_cauim}), we add a storage table for activation probabilities $ap(;)$ and then simplify the complex greedy selection (Equation~\ref{greedy_alg}) into more efficient derivative operations (Equation~\ref{dev}). In addition, we  we transform $ap(;)$ into continuous values closing to $0,1$ to signify the activated states $T_i$ of each node on average. And by this procedure, we obtain $\widehat{\mathbb{E}{{\tau}_{i}}}$ which is the approximation of the expectation on unobserved ${\tau}_i$.   }}
    \label{alg-gd}
\end{figure}


\paragraph{A-CauIM.}Four components are presented in Figure~\ref{alg-gd}. Initialization is to obtain a trained $\widehat{ITE}(;\theta^{opt})$, the same as that of G-CauIM. In the diffusion process, our objective is to calculate  $ap(\cdot;)$. To address the inefficiency loops of G-CauIM, we enhance the computation by utilizing Torch for fast graph computation. In sum, during each round of seed selection, we initially utilize a bipartite graph, where two sides are constructed by hyperegdges and nodes, to approximately compute the activated probability $ap(v_i;S)$ of each node $v_i$ under $S$ iteratively. Specifically, during $j$-th iteration of propagation (given current seed set $S$), according to the SICP model mentioned in Section~\ref{preliminary}, suppose that hyperedge $h_p$ is chosen, and its internal node $v_q$ is activated. The activation probability of $h_p$ and $v_q$ in $j$-th iteration are computed as $Prh({h_p,j}) $ and $Prv({v_q,j})$, respectively, as defined in Equation ~\ref{prob-cal}. Here $\mathcal{H}_q$ denotes the set of hyperedges containing $v_q$, $P_{SICP}$\footnote{Here $P_{SICP}$ could be seen as the same constant for each node~\cite{xie2022influence, tang2015influence} and it is easily extended to the cases where nodes are attributed to different activation probability.} represents the basic activation probability of a node. In this sense, we derive $ap( v_q; S) = \lim_{n\rightarrow +\infty} Prv({v_q,n})$. This process provides a rapid approximation of the original multiple randomizations of propagation. 
\begin{equation}
	\begin{aligned}
&Prh({h_p,j}) = \sum_{v_k \in h{p}}{Prv({v_k,j-1})}/{|\mathcal{H}_k|},\\
&Prv({v_q,j}) = 1-\Pi_{h_p \in \mathcal{H}_q}(1-P_{SICP})^{Prh({h_p,j})}
	\end{aligned} \label{prob-cal}
\end{equation}

The next step is ITE estimation. We convert $ap( v_i; S)$  into values $1-e^{C[ap(.;S)-C_d]_+}$ to represent the activated states $T_i$ of each node on average according to the obtained probabilities. Here $C, C_d$ are a priori constants. It aims to be close to binary treatment $0,1$ to fit the $\widehat{ITE}$ model and maintain the differentiability. Using $ITE(;\theta^{opt})$, we determine $\widehat{\mathbb{E} \tau_{i}}$. For the Greedy selection process, we calculate the marginal gain using $\sigma(S) = \sum_{v_i \in \mathscr{V}} ap(v_i, S)\widehat{\mathbb{E} \tau_{i}}.$ to approximate $\mathbb{E}\left[\sum_{v_i \in \mathscr{V}} ap(v_i, S){\tau_{i}}\right].$ aforementioned in Section~\ref{formulation}. Ideally, we hope A-CauIM would utilize Torch to differentiate the objective function (Equation~\ref{greedy_alg}) to identify the node that maximizes marginal gains. However, it is usually unreliable since the indices of nodes are discrete and not amenable to differentiation. To address this issue,  we use the asymptotic approximate version of $ap(v_q; S)$ (which is "parameter-based continuous") to replace the indicator of ${0, 1}$ on each node. On this basis, Equation~\ref{greedy_alg} has been transformed into an operation that is continuously differentiable by Torch, despite the cost of losing some information in unused diffusion process.
\begin{equation}
	\begin{aligned}
		 v := \arg\max_{v\notin S}\Big\{\dfrac{\partial\left({\sigma(S )}\right)  }{\partial\left( ap(v; S)\right)}*ap(v; S) \Big\}.
	\end{aligned} \label{dev}
\end{equation}


Finally, in each round, we select the node with the highest derivative value multiplied by its activated probability as the seed node in Equation~\ref{dev}, which exerts the greatest impact on ${\sigma(S)}$ under a small perturbation to the connection probability between $v$ and $S$ in the whole propagation. Importantly, we have reduced the complexity of this problem in our settings from $O(KRnm)$ to $O(Km\mathbb{E}_{h \in \mathscr{H}}|h|)$, where $m,n$ is the number of hyperedges and nodes , $K$ is seed set number identified in our preliminaries, and $R$ is simulation number of propagation process (The complexity of the ITE estimation module, utilized in both algorithms, is excluded here). Such improvement is especially significant on relatively sparse graphs.

\section{Theoretical analysis}\label{theory}

In this section, we first prove that the traditional greedy algorithm could be naturally extended to the hypergraph and remains the approximate optimal guarantee. Then we demonstrate that this approximate optimal guarantee still holds for our {CauIM} algorithm (\textbf{Theorem.~\ref{theorem_main}, challenge 2}). In addition, we show that CauIM's performance is robust to the estimation error of ITE (\textbf{Theorem.~\ref{robust_thm}, challenge 1 and challenge 3}). \footnote{\citet{antelmi2021social, zheng2019non} claimed their hypergraph does not contain submodularity. However, their hypergraph is special. They define the hyperedge as $(H , t)$\footnote{\citet{antelmi2021social} considered the special directed hypergraph. In $(H, t)$, $H$ is the set of nodes, and $t$ is the single tail node.\iffalse \zhiheng{revised}\fi}. Further, \citet{gangal2016hemi} demonstrated the submodularity of the general hypergraph. \cite{wang2016efficient} proposed the new weighted influence maximization problem. However, their attributes corresponding to each node is a priori assumed to be negative, which is different from general ITE (ITE can be positive or negative). Moreover, \citet{erkol2022effective} stated that the submodularity on the temporal network might not be held.}

\begin{proposition}\label{prop_np}
Our {CauIM} problem model is NP-hard.
\end{proposition}

We refer readers to Appendix.~\ref{proof_np}. Due to the Influence Maximization (IM) problem itself being NP-hard, our CauIM can be naturally reduced to the traditional IM problem ($\tau = 1$) and is therefore also NP-hard and more complex than the traditional problem.

\begin{lemma}[Approximal optimal guarantee of greedy IM on hypergraph]\label{appro_hyper_lemma}
The Greedy method on the \textbf{hypergraph}  can achieve the $(1-\frac{1}{e})$ approximate optimal guarantee.
\end{lemma}
\begin{proof}[Proof of sketch]
It is equivalent to consider CauIM with $\tau_i = 1, i\in \mathscr{V}$. This approximate optimal guarantee is due to two elegant properties: 1) monotonicity and 2) submodularity. Firstly, according to $ \sigma(S_0 \cup v) - \sigma(S_0) \geq 0$ when $\tau_i = 1$, the monotonicity natrually holds. Secondly, we consider the submodularity in the hypergraph. This part has been proved by~\cite{gangal2016hemi}, where they constructed an argumented graph $V^{aug} = \{\mathscr{V} \cup \mathscr{H}, E\}$. Here the edge is indentified as $ e:=(v, h), v \in \mathscr{V}, h \in \mathscr{H}, v \in h$. Then conduct the same submodularity analysis as in the traditional graph.

\end{proof}

\begin{condition}[boundedness increasement of reachable probability]\label{bound_probability}
We define the compherehensive reachable probability from a set (node) $v_1 \in \mathscr{V}$ to a set (node) $v_2 \in \mathscr{V}$ as $p_{v_1 v_2}$. $\forall v_1 \subseteq v_2 \subseteq V, |v_2| = |v_1|+1$, we have the bounded condition of the increasement of the reachable probability: $\forall v^{\prime} \in \mathscr{V}, |p_{v^{\prime} v_1} - p_{v^{\prime} v_2}| \leq \epsilon_1$.  Moreover, $\max_{v \in \mathscr{V}}\sum_{v_i \in R(v)} |\tau_{v_i}| p_{v v_i} \leq \epsilon_2$, where $R(v)$ denotes the successors that $v$ can arrive within $n$ steps. Here $\epsilon_1, \epsilon_2$ are both a priori constants.
\end{condition}

Notice that this condition is fairly broad and model-free.

\begin{theorem}[Approximate optimal guarantee of CauIM]\label{theorem_main}
1) If $\tau_i >0, i \in \mathscr{V}$, the {CauIM} algorithm can achieve the $(1-\frac{1}{e})$ optimal approximate guarantee. 2) If we do not have $\tau>0$, then a more generalized guarantee is $\sigma(S^{g}_{K}) \geq (1-\frac{1}{e})\left(\sigma(S^*) - K\epsilon_1 \epsilon_2\right) - \epsilon_2 e^{\frac{1}{K}-1}.$   
\end{theorem}

\begin{proof}[Proof of sketch]
We defer the detailed proof in Appendix.~\ref{proof_theorem_main}. The first reuslt based on $\tau_i>0$ naturally holds. It is because the monotonocity and submodularity still hold when each node is attributed with non-negetive weights (\emph{i.e.}, ITE). We focus on the second results, when $\tau_i > 0$ is not guaranteed, these two important properties will not further hold. We derive new technology to a generalized version of weak monotonicity and weak submodularity.

\end{proof}

The estimation error of $\sigma(\cdot)$ can be traced back to both the Monto-Carlo strategy and the estimation error of $\tau_i$ during representation learning. We summarize it as the result on robustness analysis as follows.

\begin{theorem}[{Robustness}]\label{robust_thm}
We denote the MC estimation of $\sigma(S)$ as $\hat{\sigma(S)}$. If $\forall S \subseteq \mathscr{V}, |\frac{\hat{\sigma}(S)}{\sigma(S)} \in [1-\gamma, 1+\gamma]$ and $
\gamma \leq \frac{\varepsilon / k}{2+\varepsilon / k}, \gamma>0
$, then our {CauIM} problem can achieve the optimal guarantee can be transfered to 
$\sigma(S^{g}_{K}) \geq (1-\frac{1}{e}-\epsilon)\left(\sigma(S^*) - K\epsilon_1 \epsilon_2\right) - \epsilon_2 e^{\frac{1}{K}-1}$.

\end{theorem}

We refer readers to Appendix.~\ref{proof_robust_thm} for details. In addition, notice that the estimation error $\hat{\tau_i} - \tau_i$ also casuses the error  $\hat{\sigma}(S) - \sigma(S)$. Specifically, if we have $|\hat{\tau_i} - \tau_i |\leq \delta$, then $\hat{\sigma}(S) - \sigma(S) $ can also be bounded.

\begin{corollary}[Robustness of noise]\label{robust_coro}
We consider the ideal CauIM case without MC strategy. The traditional IM objective function (\emph{i.e., $\tau _i = 1, \forall v_i \in \mathscr{V}$}) is denoted as $\sigma_{naive}$. If $|\hat{\tau}_i - \tau_i |\leq \delta$ and $|\frac{\sigma_{naive}(S)}{\sigma(S)}|\leq \frac{\gamma}{\delta}$, then Theorem.~\ref{robust_thm} holds.
\end{corollary}

We defer the proof in Appendix.~\ref{proof_robust_coro}.

\section{Experiments}\label{experiments}

 We execute experiments on four datasets and validate the findings presented in Theoretical Analysis Section. The anonymous code of implementation is available in the link~\footnote{
 {\url{}{https://anonymous.4open.science/r/CauIM-0331}}}.
  We aim to answer the following three questions.


 \noindent $\bullet$  \textbf{RQ1: Effectiveness} (Theorem~\ref{theorem_main}) To calculate the max sum of node ITE that represents the overall individual gains, can our G-CauIM and A-CauIM outperform the traditional IM methods and maintains high efficiency?\\
 $\bullet$  \textbf{RQ2: Robustness}~(Theorem~\ref{robust_thm},~Corollary~\ref{robust_coro})  If our ITE estimation is not accurate enough, can CauIM perform more robustly, namely, achieves an approximate result close to the normal state?\\
 $\bullet$  \textbf{RQ3: Sensitivity} Which part/parameters of the model are essential for the performance of CauIM? The effectiveness of different components of the model are explored.
    
\paragraph{Experimental Settings}
Our real-world data comprises three real-world public datasets: GoodReads~\footnote{\url{https://www.goodreads.com/}{https://www.goodreads.com/}}, Contact~\citep{mastrandrea2015contact}, and Email-Eu~\citep{Benson-2018-simplicial}. Furthermore, we incorporate a synthetic dataset named SD-100 comprising 100 nodes and 100 hyperedges, whose initial treatments and feature settings are detailed in Appendix~\ref{data_ass}. We compare the performance of G-CauIM and A-CauIM with the traditional greedy selection method without parameter adjusting strategy (Noted as "Baseline") on the aforementioned four datasets. Additionally, we choose the randomized seed selection strategy as another baseline.  We randomly conduct each experiment for $10$ times and each time for $20$ rounds in influence estimation. For basic hyperparameters, we set seed number $K=15$ and spread probability $P_{SICP}$ as 0.01 ( denoted as $p$ for simplicity). The evaluation metric is the sum of ITE spread by selected seeds with the same trained ITE estimation module illustrated in Section ~\ref{methodology}. We refer readers to Appendix~\ref{data_ass} for detailed illustrations of datasets and basic assumptions.


\subsection{General CauIM Performance}
\textbf{RQ1}~
Results of performance experiments are presented in Table ~\ref{tab:performance_comparison} and a detailed analysis of two datasets is shown in Figure~\ref{comp1} and Figure~\ref{comp2}. We can summarize it as four phenomenons: 1) G-CauIM shows slightly improvement over traditional Greedy and improves significantly compared with random selection, while the gap of their curves widens as the seed number increases. 2) A-CauIM achieves comparable performance with G-CauIM while significantly enhancing efficiency.  3) The fluctuation amplitude in the G-CauIM curve is relatively small, as demonstrated in Figure~\ref{tab:performance_comparison}, due to its alignment with the dynamic changes in individual effects and support from our Theorem~\ref{theorem_main} and Theorem~\ref{robust_thm} in Theoretical Analysis Section. 4) Traditional Greedy loses its advantage in most situations and degenerates to the random method. More details can be found in Appendix~\ref{expd}.

\begin{table}[t]
\setlength{\tabcolsep}{1mm}
\centering
    \caption{RQ1: Performance comparison of four different methods under four datasets (seed number=15). Our methods gain general improvements compared with baselines: Traditional Greedy (donated as "Baseline") and Random Selection). }
    \label{tab:performance_comparison}
    \centering
        {\begin{tabular}{l|l|l|l|l}
            \hline
            \textbf{Methods} & \textbf{GoodReads}& \textbf{Contact} & \textbf{Email-Eu}& \textbf{SD-100} \\
            \hline
            Baseline & 297.56& 68.12 & 735.28 & 138.91 \\
            Random & 45.86& 66.51 & 590.67 & 145.97 \\
           G-CauIM & \textbf{330.25}& \textbf{69.53} & \textbf{804.28} & 151.59 \\
           A-CauIM & 302.17& 66.78 & 802.41 & \textbf{160.49} \\
            \hline
        \end{tabular}}
\end{table}

\begin{table}[t]
\setlength{\tabcolsep}{1mm}
    \centering
    \caption{G-CauIM presents a significant efficiency improvement compared with other competitive baselines (GPU, Torch $1.11.0$). The complexity  is analyzed in Section~\ref{methodology}. }
    \label{tab:performance_time}
    \centering
        {\begin{tabular}{l|l|l|l|l}
            \hline
            \textbf{Methods} & \textbf{GoodReads}& \textbf{Contact} & \textbf{Email-Eu}& \textbf{SD-100} \\
            \hline
              Baseline & 1day & 6h40min & 22h & 3h\\
            Random  & 2s& 2s & 2s & 2s \\
            G-CauIM & 1day02h & 7h & 23h &  3h\\
           A-CauIM & 28s& 115s & 550s & 53s \\
            \hline
        \end{tabular}}
\end{table}

\begin{figure}[t]
\subfigure[Result on GoodReads]{
    \includegraphics[width=0.23\columnwidth]{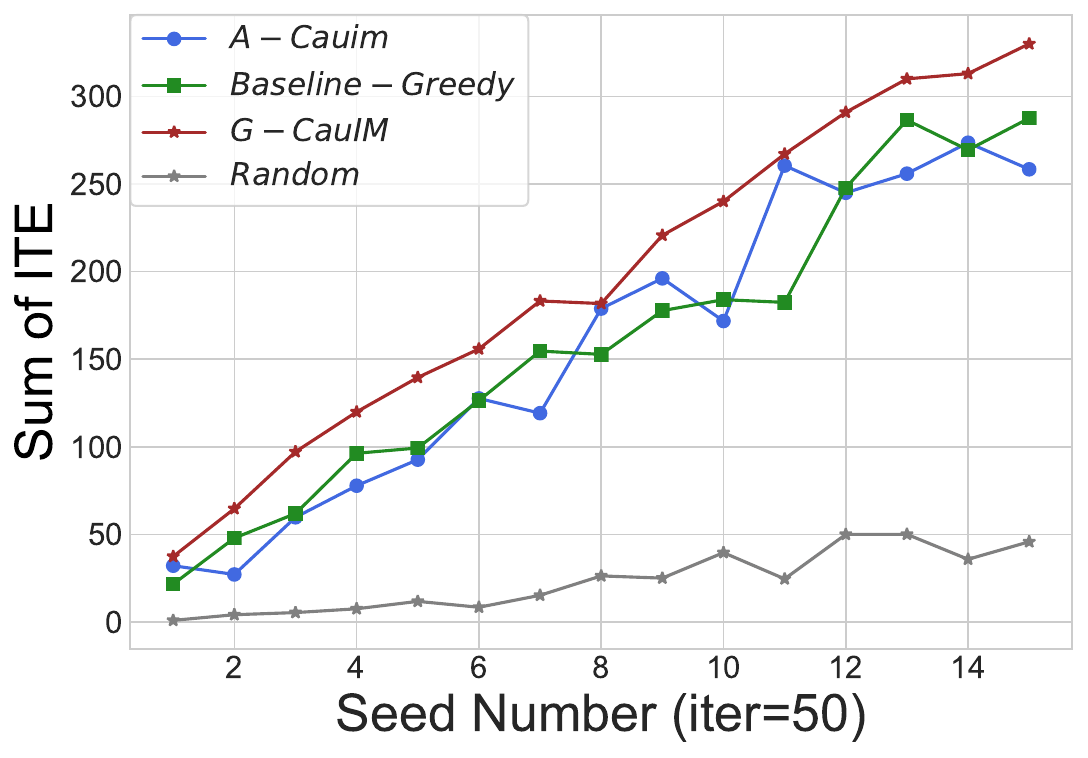}
    \label{comp1}
}
\subfigure[Result on Email-Eu]{
    \includegraphics[width=0.23\columnwidth]{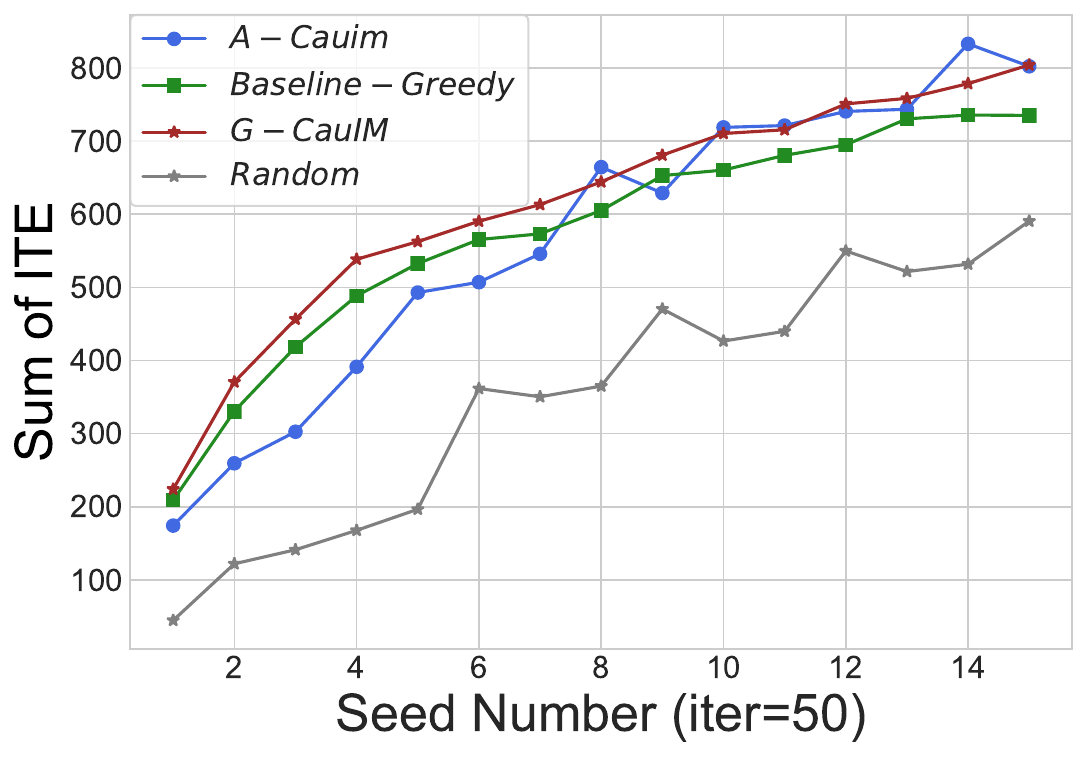}
    \label{comp2}

}
\subfigure[Analysis on noise $\epsilon$] {
    \includegraphics[width=0.23\columnwidth]{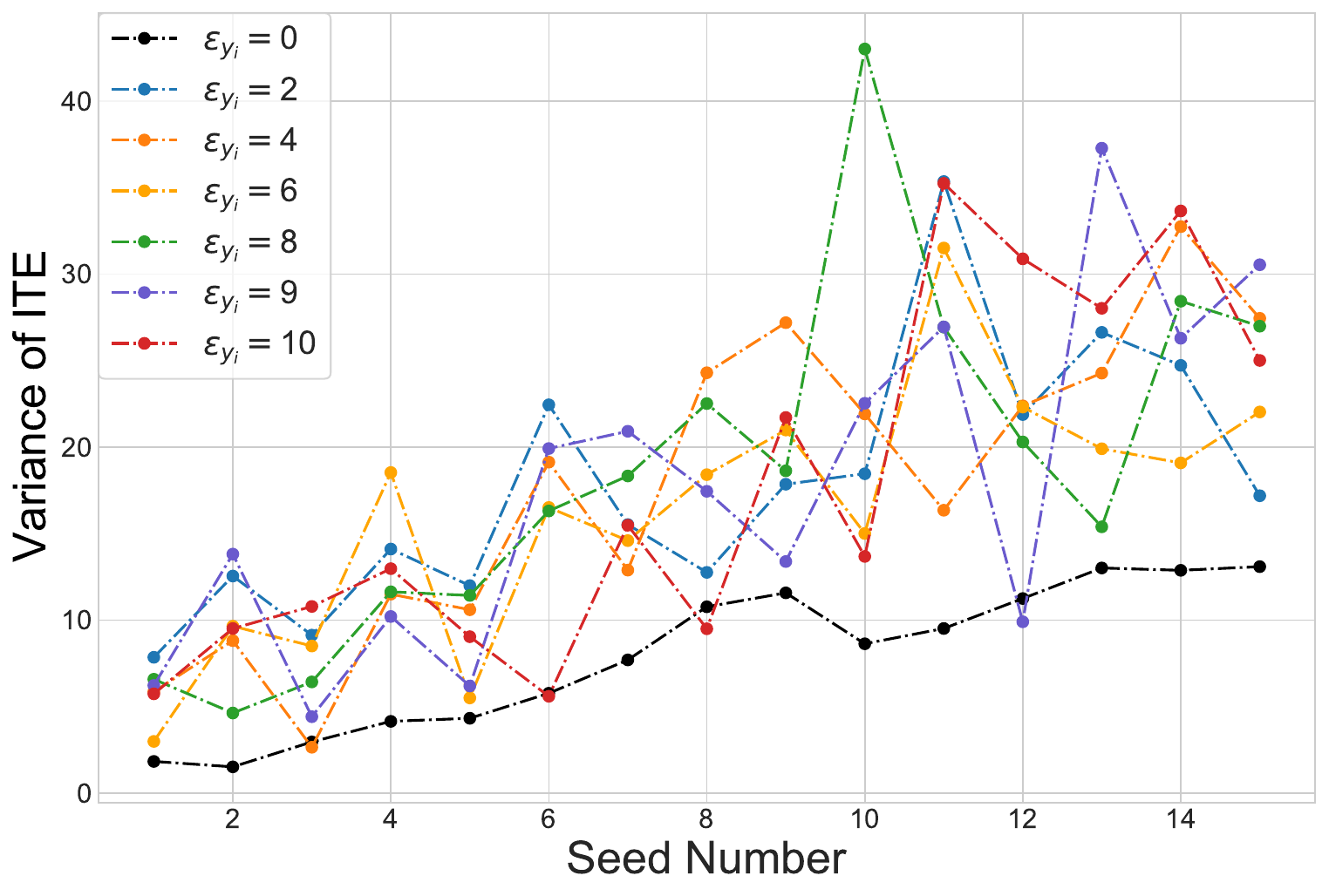}
		\label{robust}
	}
	\subfigure[Analysis on various $p$] {
		\includegraphics[width=0.23\columnwidth]{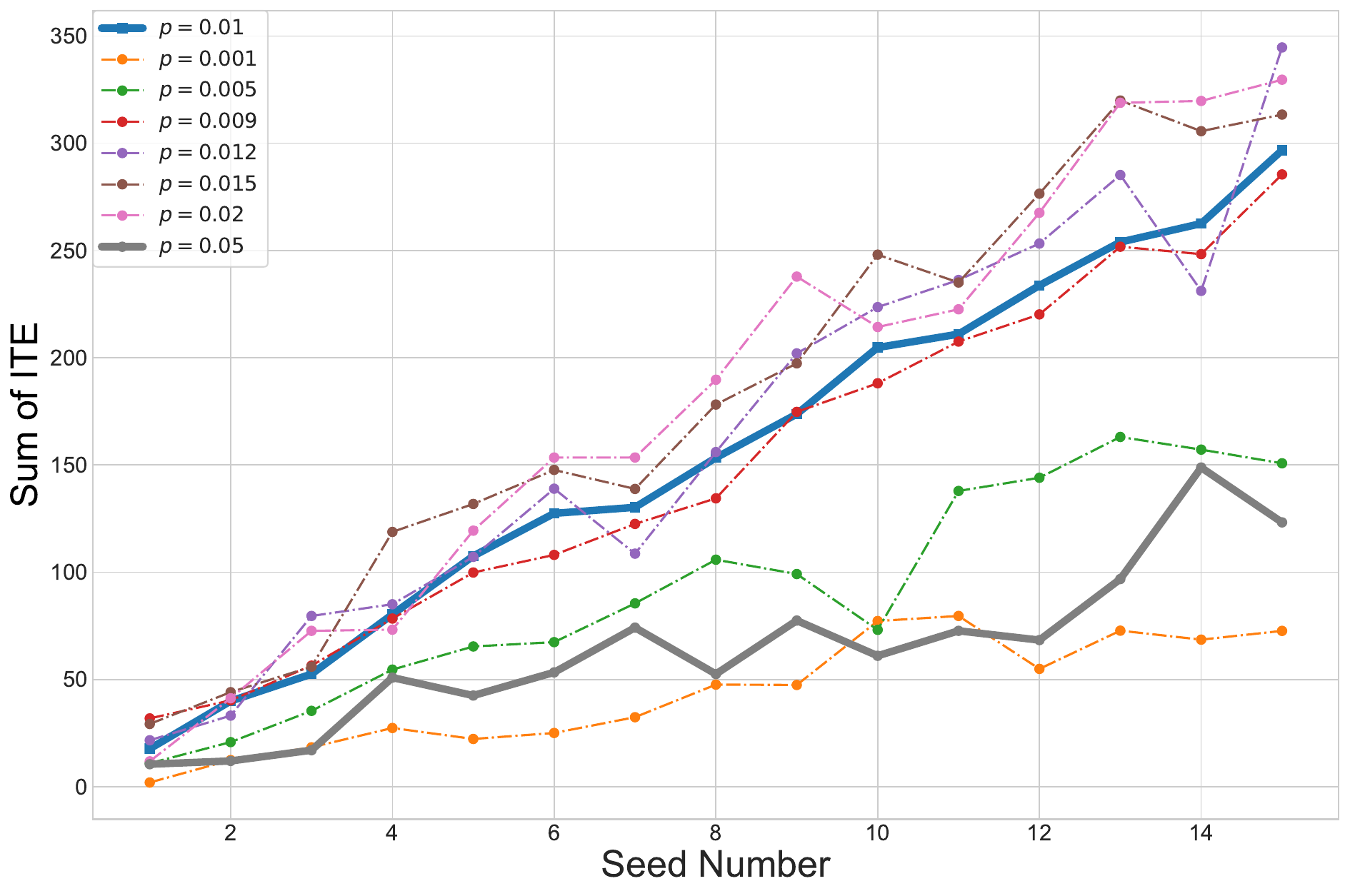}
		\label{p}
	}

 \caption{a \& b)~Performance of {CauIM} on the GoodReads and Contact dataset. ``Iter'' refers to the time step in each seed selection round here. c) Variance curve trend with different noise in individual ITE. $\epsilon_{y_i}$ represents the standard variance of the noise included. d) Final sum of ITE under various $p$.}

\end{figure}

\subsection{Robustness and Sensitivity Analysis }
Experiments for addressing RQ2 and RQ3 are conducted on A-CauIM. It is sufficient since A-CauIM is a learning-based method with slightly more uncertainty but outstanding efficiency. We conduct the experiments on the GoodReads dataset for ease of analysis.

\noindent \textbf{RQ2 \& RQ3} ~To examine robustness of our model, we add noise $\epsilon_{y}$ to individual ITE simulation results , where $\epsilon_{y} \sim N(0,\sigma^2)$. Then, the randomness of experiments comes from three parts: 1) $\epsilon_y$ , 2) the propagation probability $p$ , and 3) the dynamical ITE of each node.  We modify the scale of the noise and plot corresponding curves. Figure~\ref{robust} can provide the following illustrations: 1) when the noise becomes larger (not exceeding $9$), it is not enough to counter the random variance of the propagation process and dynamical ITE itself. This powerlessness phenomenon is changed until noise reaches $10$. ~2) In the process of seed growth, when the number of seeds is smaller than $10$ in the early stage, the effect of total noise of the infected nodes is not large. Still, after the seeds increases over $10$, the scale of affected nodes rises significantly, for the effect of the sum of noise becomes large. Thus, we conclude that our model can maintain a relatively robust state with the $\sigma$ of noise not exceed $8$. This result is consistent with our theoretical part.  For detailed explanations, please refer to Appendix \ref{expd}.

We detect how the components of the algorithm work and investigate how the diffusion process affects A-CauIM with various parameter $p$; As shown in figure~\ref{p}, curves of ITE sum nearly merge to one when $p$ is larger than 0.01 and smaller than 0.05. This convergence indicates that complete traversal of nodes in the hypergraph occurs with a sufficiently high propagation probability, leading to the stabilization of the total dynamic ITE. Additionally, it is noteworthy that the performance sharply declines when $p=0.05$, possibly due to the broader diffusion process intensifying the randomness in dynamical ITE and consequently augmenting the uncertainty of influence spreading.  In conclusion, with p changing in a certain range, our algorithm will work stably well, and $p=0.02$ is the approximately best choice for achieving the best performance.





\section{Related work}
\label{related-work}

\noindent{\textbf{Influence Maximization (IM)}} IM is first identified as an algorithm problem by Kempe in~\cite{kempe2003maximizing} and has given rise to several notable variants, including simulation-based (CELF~\citep{leskovec2007cost}), sketch-based (RIS~\citep{borgs2014maximizing}, TIM~\citep{tang2014influence}, IMM~\citep{tang2015influence}), and heuristic algorithms (HDD~\cite{xie2022influence}), etc. Three key elements of the problem are 1) {graph structure, 2) diffusion process and 3) seed selection}. For instance, CELF, RIS, and TIM emphasize iterative algorithmic enhancements while ensuring theoretical assurances within triggering models. In contrast, the latest heuristic algorithm, HDD, prioritizes computational efficiency at the cost of sacrificing theoretical guarantees. Additionally, recent learning-based IM methods~\citep{chen2023touplegdd,kumar2022influence,ling2023deep}focus on understanding the inherent nature of individual node representations concerning the marginal influence gain. Nevertheless, these methods exhibit limitations in model generalization and the reliability of final results.

The original optimization objective, which entails the summation of activated nodes, requires reassessment in diverse scenarios. We refer readers to the Appendix~\ref{related} for more details. Researchers have increasingly focused on exploring the heterogeneity of importance between nodes based on these methods~\citep{wang2016efficient, hu2023triangular}. However, much of the existing literature introduces prior information based on fixed topological features or node attributes and assumes each individual importance to be a constant without integrating a unified computational framework involving environmental sensitivity for calculating the dynamic individual effect of nodes.

\noindent \textbf{Treatment effect estimation}
How to recover the ITE directly from the observational data instead of randomization test~\citep{rubin1978bayesian} is currently receiving a lot of attention. There are two main strategies for estimation: 1) weighting-based methods~\citep{li2019propensity,li2018balancing, franklin2014metrics}, and 2) representation-based methods~\citep{shalit2017estimating, johansson2020generalization}. In this paper, we follow the second strategy. 
 \citet{ma2022learning} estimated the causal effect via representation learning on the more general hypergraph. However, these methods do not consider the IM question. 

\section{Conclusion and Discussion}\label{conclusion}
In this paper, we analyze traditional IM from a causality perspective. Our CauIM framework can extract approximately optimal seed sets to achieve novel influence maximization.

\paragraph{Discussion on Sketch-based Models} In addition to the limitations inherent in triggering-based diffusion models, sketch-based models may face challenges in capturing dynamically changing node weights within their environments during the seed selection stage, when compared with greedy-based models. Furthermore, sparse graph structures could reduce the efficiency of sketch-based models, our acceleration method can be adopted to diverse graph configurations. To summarize, addressing the aforementioned challenges is a new area of focus and will be the subject of our future work.

\paragraph{Future Directions} Our paper has also opened up innovative research directions. (i) Widespread applicability; it should be noticed that the preference for greedy-based acceleration algorithms initially arises from their superior adaptability to diverse models and varying node weights compared to other traditional IM methods. (ii) Scenario transferability; ongoing exploration involves examining the intricate interference of online IM methods in various type of incomplete networks. 


\clearpage
\bibliography{references}

\bibliographystyle{plainnat}
\appendix
\onecolumn

\section{Methodology Details}
\begin{figure}[h]
    \centering
    \includegraphics[width = 9cm]{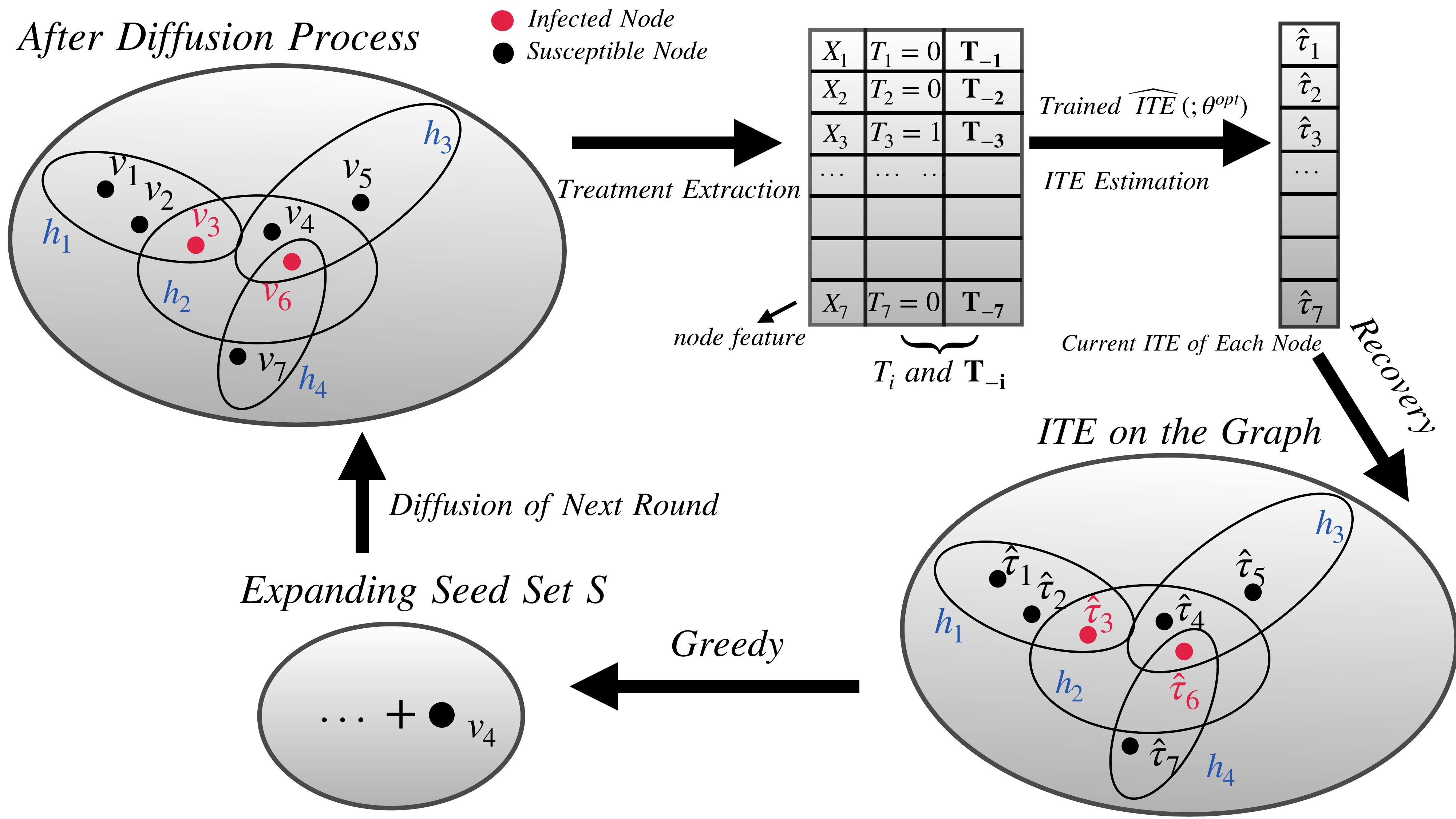}
    \caption{{The procedure of G-CauIM. $T_i$ indicates whether the node is activated or not (we can only observe one situation for each node ${v_i}$) and \bm{$T_{-i}$} represents activated status of its surroundings, as illustrated in Equation \ref{causal_effect}. For each round, we construct the ITE estimation $\hat{\tau}_{i}:= \widehat{ITE}(;\theta^{opt})$ mentioned in Section \ref{methodology} and then treat it as the node weight. Furthermore, we conduct a weighted greedy algorithm with SICP propagation mechanism (Figure~\ref{pic1}). As a result, we expand the seed set ($v_4$ is added), targeting the (estimated) largest sum of ITE. {The main challenge is that $\hat{\tau}_{i}$ are not constants, since the omitted parameter (Equation~\ref{causal_effect}) changes according to the different activation status of nodes in each iteration.} }}
    \label{alg}
\end{figure}

\section{Comparison}
\begin{table*}[htpb] 
	\centering
	\begin{tabularx}{\textwidth}{lXXX}
		\toprule
		  & Simulation-based& Proxy-based &Sketch-based  \\
		\midrule
	Basic idea
		&Use Monte-Carlo (MC) simulation to evaluate ITE influence spread    & Design proxy models to approximate influence function with varying ITE   & Construct theoretically grounded sketches considering the sum of ITE under the diffusion model, to improve theoretical efficiency of the simulation-based methods while preserving the approximation guarantee		\\
		\midrule
		property  & NP-hard complexity, total theoretical guarantee  & Polynomial/linear complexity, no theoretical guarantee    & Quasi-linear  complexity, partial theoretical guarantee   \\
		\midrule
		Disadvantages   & Computational overheads  & Insensitive to the unstable scenarios  & Not general to a wider range of diffusion models\\
		\midrule
		Datasets &Small to medium-sized datasets under all propagation models   & Large-scale datasets with distinctive graph structure under specific propagation models  & Data sets with less compositional sparsity under Triggering propagation models \\
		\midrule
		Examples&Greedy, CELF (famous)  & HDD~\citep{xie2022influence}, HSD~\citep{xie2022influence},  CIA~\citep{zhang2023influence} & RIS~\citep{borgs2014maximizing}, IMM~\citep{tang2015influence}, BKRIS
		\\
		\bottomrule 
	\end{tabularx}%
	\caption{Comparison of three types of traditional IM  under Our CauIM framework.}
	\label{table_compare_1}
\end{table*}%


\begin{table*}[t]
	\centering
	\begin{tabularx}{\textwidth}{lXX}
		\toprule
		 &  {CauIM} &  Traditonal IM  \\
		\midrule
	Basic idea
		&Leverages observational data to estimate the ITE of each node and to maximize the sum of varying ITEs 
among the infected individuals considering environmental information 
		& Maximize the numbers of infected individuals		\\
		\midrule
		Objective function  & $\arg \max _{S \subseteq \mathscr{V} \wedge|S|=K} \mathbb{E}[|\Phi(S)|]$, where $\Phi(S)=\sum_{{v_i} \in S} \mathbb{E}\tau_{i}$ measures the sum of ITE (informal, parameter of interference is omitted)  &$\arg \max _{S \subseteq \mathscr{V} \wedge|S|=K} \mathbb{E}[|\Phi(S)|]$, where $\Phi(S)$ measures the sum of infected numbers\\
		\midrule
		Application Scenarios & Maximize total sum of individual gains & Maximize infected numbers \\
		\bottomrule  
	\end{tabularx}%

	\caption{A comparison between CauIM and traditional IM (Notice that CauIM is more challenging than the sum-weighted IM, since the ITE would be negative and varying during different propagation.)}
	\label{table_compare_2}
\end{table*}%
Further comparison is provided in Table~\ref{table_compare_1} and Table~\ref{table_compare_2}. As shown in Table~\ref{table_compare_1}, the Simulation-based CauIM is successfully implemented in this work.

\section{Experiments Details} \label{data_ass}

\subsection{Details of Datasets and Problem Background}
GoodReads~\citep{hu2008collaborative,wan2018item} collects information on different categories of books, with each item containing the book title, content, and other details. Using the "Author-book" relationship, a node represents the specific book category  and  hyperedge aggregates books written by  the same author in our hypergraph. We consider a scenario of recommending book sales, where the diffusion process is facilitated by reader groups associated with each book category node in the “Author-book” network. Treatment denotes the recommendation for book sales,  with $t_i$  set to $1$ when book node $i$ is recommended. Each book category is associated with potential sales income, influenced by recommendations and other books within the same hyperedge.  Our goal is to identify a k-set(or k kinds) of books to sell at the beginning, aimed at maximizing the total gap of sales  with/without recommendation. 
Our core optimization function is to maximize the sum of ITE, where the ITE of each book means the difference in  potential sold income with/without the recommendation in our experiment. For Contact~\citep{benson2018simplicial}, it constructs simplicial complexes by grouping individuals in close connectivity at the same timestamp, represented by  hyperedges. Here we simulate a situation where we deliver an AIDS Awareness Talks to particular students, and the core concepts can be disseminated through hyperegde groups. Our objective is to select initial student representatives to be educated in order to maximize the overall benefits of the talk (This can also be viewed as maximizing the total sum difference between having the anti-drug talk and not having it). We simulate covariates($x_i$) of each student in a Mixed Gaussian distribution considering differences among diverse groups:
\begin{equation}
    x_i \sim \sum_{j =1}^L \omega_i \mathcal{N}(\mu_j, I).
    \label{simu_x}
\end{equation}
Here we set $L=4, \omega_1 = 0.4$. Moreover, $\{\omega_1, \omega_2, \omega_3, \omega_4\} = \{0.4, 0.2, 0.1,0.3\}$, $\{\mu_1, \mu_2, \mu_3, \mu_4\} = \{0.2, -0.25, -0.3, 0.5\}$. 
Email dataset shares the similar scenario with Contact.
The ratio of nodes to hyperedges is different: For example, Goodreads is $36$ and Contact is $0.14$. The simulation of ITE  and basic treatment settings for both datasets follow~\cite{ma2022learning}.  It is worth noting that the ratio of nodes to hyperedges is contrasting between the two datasets, yielding sparse and relatively dense graphs, respectively. This exemplifies the diversity of our data selection, allowing for a more robust evaluation of algorithmic performance. Synthetic dataset SD generates initial treatments $T_i\sim$ Bernoulli$(r_0)$, where $r_0$ denotes an average affected ratio of nodes (which can be calculated easily through propagation simulations). Its covariates($x_i$) are simulated following Equation \ref{simu_x}, with distinct values of $\{\omega_1, \omega_2, \omega_3, \omega_4\} = \{0.4, 0.25, 0.15,0.2\}$, representing different group ratio.

\subsection{Supplementary Descriptions of Basic Assumptions}
\label{expd}
 The basic assumptions of the book-selling scenario should be satisfied: 1) Books and authors are many-to-many relationships; 2) the number of each kind of books sold initially is the same (we will take it as our future research topic if not the same); 3) Readers can learn about other books simply from authors of owned books. 4) Temporal change of each hyperedge is not considered. Assumptions for the other example are similar with the exception that rule 3 is replaced by ``Students can learn about core ideas simply from other students.'' Our diffusion framework is established by hypergraph-based IC as described in ~Section~\ref{methodology}. Within this model, the spread probability represents the willingness of readers to purchase the next book in the GoodReads dataset and the likelihood of core concepts spreading among students in the Contact or Email dataset.

\subsection{Supplementary Descriptions of Parameter Settings}
   Our experiments are conducted on Linux operating system with Python 3.10.14, torch 2.1.

\noindent \textbf{Exp2 for RQ2} Aimed at indicating the volatility, we modify the scale of the noise to approximate instability degree and  take on 10 realizations of each experiment,  calculating the standard variance of ITE performance among those realizations as the final result. We conduct 20 groups of experiments with $\sigma$ varying between 0 to 20, and use step 2 when $\sigma$ is lower than 8 for its changes are not noticeable. While $\sigma$  is larger than 10, the curves increase too drastically. Thus we do not present in Figure~\ref{robust}.
\begin{figure}[b] 
\vskip 0.2in
\centering
	\subfigure[] {
    \includegraphics[width=0.4\columnwidth]{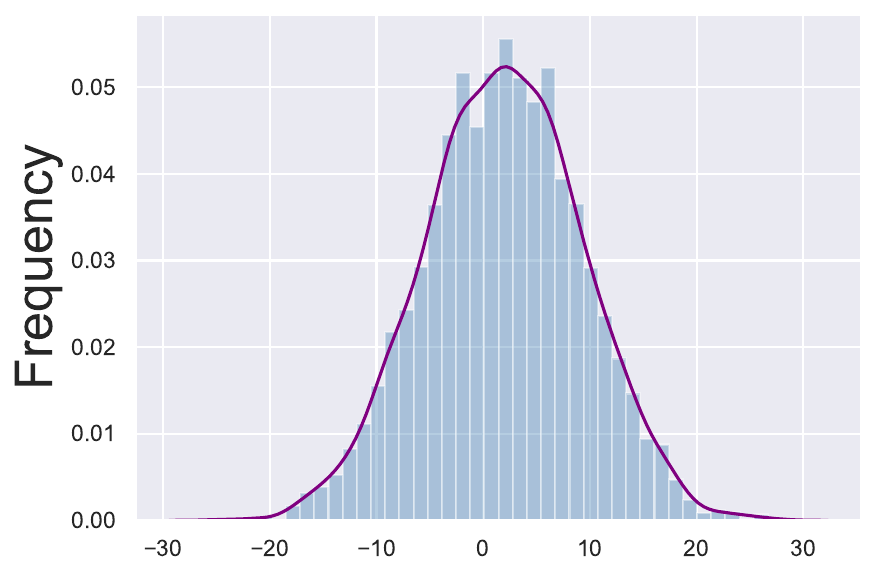}
		\label{GRITE}
	}
	\subfigure[] {
		\includegraphics[width=0.4\columnwidth]{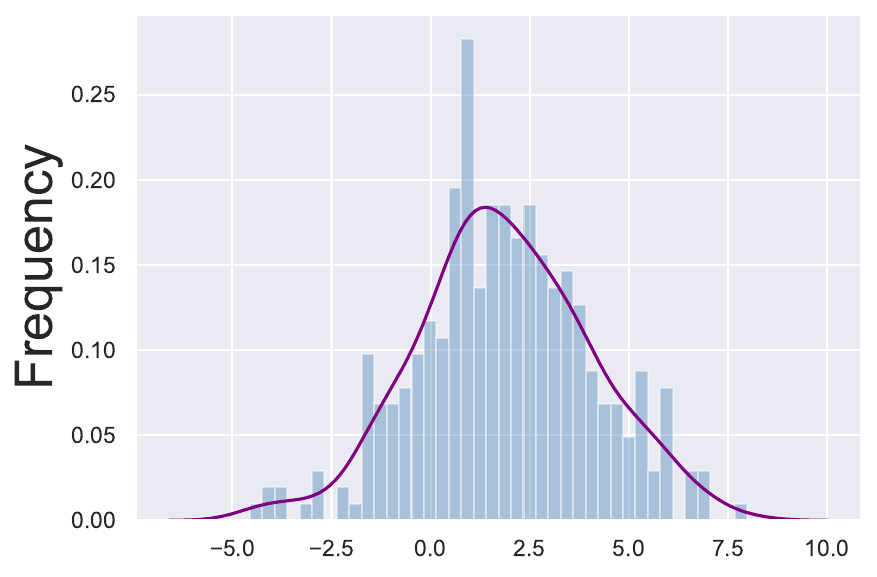}
		\label{CTITE}
	}
	\caption{a) Distribution of individual ITE in GoodReads Dataset. b) Distribution of individual ITE in Contact Dataset.}

	\label{ite_figure}
	\vskip -0.2in
\end{figure}
\section{The proof of Proposition~\ref{prop_np}} \label{proof_np}
\begin{proof}[Proof of sketch]
Notice that when $\tau_{i} = 1, \forall {v_i} \in \mathscr{V}$ and each hyperedge contains only one pair of nodes, then CauIM will degenerate to the traditional IM problem. Moreover, the IM problem with the IC diffusion model has been demonstrated to be NP-hard in \cite{kempe2003maximizing}. 

In another perspective, we can directly prove the optimal seed set of CauIM is one solution of the famous weighted set cover problem, which is well-known as NP-hard~\citep{hochbaum1996approximating}. The weighted set cover problem is equivalently defined as detecting whether there exists $k$ subsets within the total $m$ subsets, such that it can cover the universe of elements $U$. We aim to demonstrate the traditional weighted set cover problem can naturally generalize a specific instance in CauIM. We construct a bipartie graph, in which the left side denotes the total subsets $\mathscr{S} = \{\mathcal{S}_{1},\mathcal{S}_2,...\mathcal{S}_{m}\}$, while the right side denotes the elements $u \in U$, which can be seen as each node in the hypergraph. The edge between two sides is attributed with the probability $1$. Notice that in this instance, there exists a $k$-set cover if and only if there exists $k$-seed set such that could be propagated to the total $|S|+|U| = K+|U|$ nodes. Hence our CauIM is more complex and is NP-hard. 

\end{proof}

\section{The proof of Lemma~\ref{appro_hyper_lemma}} \label{appro_hyper}

\begin{proof}
It is equivalent to consider CauIM with $\tau_{i} = 1, {v_i}\in \mathscr{V}$. This approximate optimal guarantee is due to two elegant properties: 1) monotonicity and 2) submodularity. Firstly, according to $ \sigma(S_0 \cup v) - \sigma(S_0) \geq 0$ when $\tau_{i} = 1$, the monotonicity naturally holds. Secondly, we consider the submodularity in the hypergraph. This part has been proved by~\cite{gangal2016hemi}, where they constructed an augmented graph $V^{aug} = \{\mathscr{V} \cup \mathscr{H}, \mathscr{E}\}$. Here the edge $\mathscr{E}$ is composed by $ e:=(v, h), v \in \mathscr{V}, h \in \mathscr{H}$. Then conduct the same submodularity analysis as in the traditional graph, and the proof is completed.

\end{proof}

\section{The proof of Theorem~\ref{theorem_main}}\label{proof_theorem_main}

The first result based on $\tau_{i}>0$ naturally holds. It is because the monotonicity and submodularity still hold when each node is attributed with non-negative weights (\emph{i.e.}, ITE). We focus on the second result, when $\tau_{i} > 0$ is not guaranteed, these two important properties will not further hold. We derive new technology to a generalized version of weak monotonicity and weak submodularity.

\begin{proof}
To summarize, the core part is the following three claims:

\noindent \textbf{(Claim 1)}   $
    \sigma(S^{*}) \leq \sigma(S^{*} \cup S_{i}^g) + i \varepsilon_2.$
   
\noindent \textbf{(Claim 2)} $\sigma(S^{*} \cup S_{i}^g) \leq \sigma\left(S_{i+1}^g \right)-\sigma\left(S_i^g\right) + \sigma\left(S_i^g \cup S_{K-1}^*\right) + \varepsilon_1 \varepsilon_2$.

\noindent \textbf{(Claim 3)} $
    \sigma\left(S_i^g \cup S_{K-1}^*\right) \leq (K-1)[\sigma\left(S_{i+1}^g \right)-\sigma\left(S_i^g\right)] + \sigma(S_g^i) +(K-1)\varepsilon_1 \varepsilon_2.$
The optimal $K$-seed set is denoted as $S^{*} = \{s_{1}^{*}, s_{2}^{*},...s_{K}^{*}\}, S^{*}_k = \{s_{1}^{*}, s_{2}^{*},...s_{k}^{*}\}$, and the set output from our greedy {CauIM} as $S^{g} = \{s_{1}^{g}, s_{2}^{g},...s_{K}^{g}\}, S^{g}_k = \{s_{1}^{g}, s_{2}^{g},...s_{k}^{g}\}$, $k \in [K]$. Following~\cite{wang2016efficient}, notice that
\begin{equation}
    \begin{aligned}
    \sigma(S)=\sum_{u \in S}\left(\sum_{v \in V} \tau_v \cdot p_r(u, v)+\tau_u\right).
    \end{aligned}
\end{equation}

We first construct the facilitating claim to analyze the variant of the monotonicity and the submodularity property. 

\textbf{Claim 1:}   $
    \sigma(S^{*}) \leq \sigma(S^{*} \cup S_{i}^g) + i \varepsilon_2.
   \label{eqn_compare}$
    
This claim can be achieved recursively. Considering two sets $T_1 \subseteq T_2$, and an additional vertex $v \nsubseteq T_1 $, we can follow~\cite{wang2016efficient} and achieve:

\begin{equation}
    \begin{aligned}
     \left( \sigma(T_1 \cup v) - \sigma(T_1) \right) = &\sum_{v_i \in R(v)} \tau_{i} \cdot p_{v v_i}\left(1-p_{T_1, v_i}\right)  \leq \sum_{v_i \in R(v)} |\tau_{i}| p_{v v_i} = \varepsilon_2.
    \end{aligned}
\end{equation}
Then claim $1$ can be achieved by incursively taking use of it $i$ times.

\textbf{Claim 2:} $\sigma(S^{*} \cup s_{i}^g) \leq \sigma\left(S_{i+1}^g \right)-\sigma\left(S_i^g\right) + \sigma\left(S_i^g \cup S_{K-1}^*\right) + \varepsilon_1 \varepsilon_2$.

We make an extension of submodularity. Considering $S \subseteq T \subseteq \mathscr{V}$, we have
\begin{equation}
    \begin{aligned}
     &\left( \sigma(S \cup v) - \sigma(S) \right) -  \left( \sigma(T \cup v) - \sigma(T) \right)\\ =& \sum_{v_i \in R(v)} \tau_{i} \cdot p_{v v_i}\left(p_{T, v_i}-p_{S, v_i}\right) \leq \varepsilon_1 \varepsilon_2.
    \end{aligned}
\end{equation}

Hence $ \sigma(S^{*} \cup s_{i}^g) $ can be bounded as follows:

\begin{equation}
    \begin{aligned}
     &\sigma(S^{*} \cup s_{i}^g) = \sigma( S_{K-1}^{*} \cup s^{*}_{K} \cup S_{i}^g ) \\ \leq& \sigma \left(S_i^g \cup s_K^* \right)-\sigma\left(S_i^g\right)+\sigma\left(S_i^g \cup S_{K-1}^*\right) + \varepsilon_1 \varepsilon_2 \\  \leq& \sigma\left(S_{i+1}^g \right)-\sigma\left(S_i^g\right) + \sigma\left(S_i^g \cup S_{K-1}^*\right) + \varepsilon_1 \varepsilon_2. 
    \end{aligned}\label{eqn_original}
\end{equation}
 The last line is due to the selection nature of the greedy algorithm.

 \textbf{Claim 3:} $
    \sigma\left(S_i^g \cup S_{K-1}^*\right) \leq (K-1)[\sigma\left(S_{i+1}^g \right)-\sigma\left(S_i^g\right)] + \sigma(S_g^i) +(K-1)\varepsilon_1 \varepsilon_2.
    \label{eqn_original_1}$

It is due to

\begin{equation}
    \begin{aligned}
   \sigma\left(S_i^g \cup S_{k-1}^*\right)   = & \sigma \left(S_i^g \cup S_{k-2}^* \cup s_{k-1}^* \right)\\ \leq & \sigma\left(S_i^g \cup s_{k-1}^*\right)-\sigma\left(S_i^g\right) + \sigma\left(S_g^i \cup S_{k-2}^*\right) + \varepsilon_1 \varepsilon_2 \\
   \leq& \sigma\left(S_{i+1}^g \right)-\sigma\left(S_i^g\right) + \sigma\left(S_g^i \cup S_{k-2}^*\right) + \varepsilon_1 \varepsilon_2 \\
    \end{aligned}
\end{equation}

 The last line is due to the selection nature of the greedy algorithm. Therefore, recursively, we have
\begin{equation}
    \begin{aligned}
    \sigma\left(S_i^g \cup S_{K-1}^*\right) \leq (K-1)[\sigma\left(S_{i+1}^g \right)-\sigma\left(S_i^g\right)] + \sigma(S_g^i) +(k-1)\varepsilon_1 \varepsilon_2.
    \end{aligned}
\end{equation}

Combined with Equation~\ref{eqn_original}, Equation~\ref{eqn_compare} and Equation~\ref{eqn_original_1}, we have
\begin{equation}
    \begin{aligned}
    \sigma(S^{*}) \leq K [\sigma\left(S_{i+1}^g \right)-\sigma\left(S_i^g\right)] + \sigma(S_i^g) + K\varepsilon_1 \varepsilon_2 + i\varepsilon_2.
    \end{aligned}
\end{equation}
It equals to 
\begin{equation}
    \begin{aligned}
    &\sigma(S_{i+1}^g) \geq (1-\frac{1}{K})\sigma(S_{i}^g) + \frac{\sigma(S^*)}{K}-\varepsilon_1 \varepsilon_2 - \frac{i\varepsilon_2}{K}\\
    &(1-\frac{1}{K})^{k-i-1}  \sigma(S_{i+1}^g)\\ \geq& (1-\frac{1}{K})^{k-i}\sigma(S_{i}^g) + (1-\frac{1}{K})^{k-i-1} \left( \frac{\sigma(S^*)}{K}-\varepsilon_1 \varepsilon_2 - \frac{i\varepsilon_2}{K}\right).
    \end{aligned}
\end{equation}
Take the sum of $i \in \{0,2,...K-1\}$, we have
\begin{equation}
    \begin{aligned}
    \sigma(S^{g}_{K}) &\geq \sum_{i=0}^{K-1}(1-\frac{1}{K})^{k-i-1} \frac{\sigma(S^*) - K\varepsilon_1 \varepsilon_2}{K} - \sum_{i=0}^{K-1}(1-\frac{1}{K})^{k-i-1} \frac{i}{K}\varepsilon_2\\
    & \geq \left(1-\left(1-\frac{1}{k}\right)^k\right)\left(\sigma(S^*) - K\varepsilon_1 \varepsilon_2\right) - \varepsilon_2 (\frac{1}{e})^{1-\frac{1}{K}} \int_{0}^{1}xe^x dx \\
    & \geq (1-\frac{1}{e})\left(\sigma(S^*) - K\varepsilon_1 \varepsilon_2\right) - \varepsilon_2 e^{\frac{1}{K}-1}.
    \end{aligned}
\end{equation}

This lower bound is also applicable when $\tau_i \leq 0$ exists.

\end{proof}
\section{The proof of Theorem~\ref{robust_thm}} \label{proof_robust_thm}
\begin{proof} Following~\citet{chen2013information} , we have ($k \in [K]$)
\begin{equation}
\begin{aligned}
 f\left(S_i^g \cup\left\{{s}^{*}_{k}\right\}\right)
& \leq \frac{1}{1-\gamma} \hat{f}\left(S_i^g \cup\left\{\bar{s}^*_{k}\right\}\right) \\
& \leq \frac{1}{1-\gamma} \hat{f}\left(S_i^g \cup\left\{s_{i+1}\right\}\right) \\
& \leq \frac{1+\gamma}{1-\gamma} f\left(S_i^g \cup\left\{s_{i+1}\right\}\right) .
\end{aligned}
\end{equation}

Analogously, we have
\begin{equation}
\begin{aligned}
f\left(S_{i+1}^g\right) \geq \frac{1-\gamma}{1+\gamma}\left(\left(1-\frac{1}{K}\right) f\left(S_i^g\right)+\frac{f\left(S^*\right)}{K} - \varepsilon_1 \varepsilon_2 - \frac{i\varepsilon_2}{K}\right) .\end{aligned}
\end{equation}

Hence, recursively,
\begin{equation}
\begin{aligned}
f\left(S_K^g\right) 
& \geq \sum_{i=0}^{K-1}\left(\frac{(1-1 / K)(1-\gamma)}{1+\gamma}\right)^{K-i-1} \cdot \frac{1-\gamma}{(1+\gamma) K} \cdot \left[f\left(S^*\right) - K\varepsilon_1 \varepsilon_2\right] \\&- \sum_{i=0}^{K-1}(\frac{(1-1/K)(1-\gamma)}{1+\gamma})^{k-i-1} \frac{i}{K}\varepsilon_2 \\
& \geq \frac{1-\left(\frac{1-\gamma}{1+\gamma}\right)^K\left(1-\frac{1}{K}\right)^K}{(1+\gamma) K /(1-\gamma)-K+1} \left[f\left(S^*\right) - K\varepsilon_1 \varepsilon_2\right] - \varepsilon_2 e^{\frac{1}{K}-1} \\
& \geq \frac{1-\left(\frac{1-\gamma}{1+\gamma}\right)^K \cdot \frac{1}{e}}{(1+\gamma) K /(1-\gamma)-K+1} \left[f\left(S^*\right) - K\varepsilon_1 \varepsilon_2\right] -\varepsilon_2 e^{\frac{1}{K}-1} \\
& \geq \frac{1-\frac{1}{e}}{(1+\gamma) K /(1-\gamma)-K+1} \left[f\left(S^*\right) - K\varepsilon_1 \varepsilon_2\right] - \varepsilon_2 e^{\frac{1}{K}-1}\\
& \geq\left(1-\frac{1}{e}\right)\left(1-\frac{(1+\gamma) K}{1-\gamma}+K\right) \left[f\left(S^*\right) - K\varepsilon_1 \varepsilon_2\right]-\varepsilon_2 e^{\frac{1}{K}-1}\\
& \geq\left(1-\frac{1}{e}-\left(\frac{(1+\gamma) K}{1-\gamma}-K\right)\right) \left[f\left(S^*\right) - K\varepsilon_1 \varepsilon_2\right]-\varepsilon_2
e^{\frac{1}{K}-1} \\
& \geq\left(1-\frac{1}{e}-\varepsilon\right) \left[f\left(S^*\right) - K\varepsilon_1 \varepsilon_2\right] -\varepsilon_2 e^{\frac{1}{K}-1} .
\end{aligned}
\end{equation}

\end{proof}

\section{The proof of Corollary~\ref{robust_coro}} \label{proof_robust_coro}

\begin{equation}
    \begin{aligned}
    |\hat{\sigma}(S) - \sigma(S)| &= \sum_{u \in S}\left(\sum_{v \in V} \left(\hat{\tau}_v - \tau_v\right) p_r(u, v)+\hat{\tau}_v - \tau_v \right) \\
    &\leq  \frac{\gamma \sigma(S)}{\sigma_{naive}(S)} \sum_{u \in S}\left(\sum_{v \in V} p_r(u, v)+1 \right) := \gamma \sigma(S).
    \end{aligned}
\end{equation}

Hence $|\frac{\hat{\sigma}(S) }{\sigma(S)} - 1 | \leq \delta \frac{\sigma_{naive}(S)}{\sigma(S)} \leq \gamma$. We have proved.

\begin{algorithm}[t]
\caption{Monte Carlo-based greedy {CauIM}} 
\label{alg1: traditional_cauim_mc} 

\begin{algorithmic}[1] 

\REQUIRE Hypergraph $\mathscr{G}(\mathscr{V}, \mathscr{H})$, size of the seed set $K$, a constant $T$.

\STATE Initialization: $S_0 = \emptyset$, $k=0$.
\STATE \textbf{ITE recovery}.
\FOR{$|S_0| <K$}

\STATE $v_0 = \arg \max _{v \notin S_{0}}\left\{\textbf{MC}\left(S_{0} \cup\{v\}, T \right) -\textbf{MC}\left(S_{0}, T \right)\right\}.$
\STATE $S_0=S_0{} \cup\left\{v_0\right\}$

\ENDFOR
\ENSURE The deterministic seed set $S_0$ with $|S_0| = K$.

\textbf{Function MC:}
\REQUIRE Iteration $T$, current node set

\STATE count = 0
\FOR{$i \in [T]$}
\STATE We conduct the diffusion process with $T$ steps, and compute the sum of causal effects $\sigma_{T}(S_0) = \sum_{\text{node~}j \text{~is~activated}}\tau_j$. $\text{count} = \text{count} + \sigma_{T}(S_0)$.
\ENDFOR

\ENSURE Return $count/T$.

\end{algorithmic} 
\end{algorithm}

\section{Auxiliary algorithms and Additional discussions}\label{related}

\paragraph{More discussions of HGCN module}
The higher-order interference representation  $O_{i}()$ is learned according to to~\citep{ma2022learning,ma2021causal}, employing hypergraph convolution operator within HGCN module: $\mathbf{O}^{(l+1)}=\text{LeakyReLU}\left(\mathbf{LO}^{(l)}\mathbf{W}^{(l+1)}\right)$ where $\mathbf{L}$ denotes Laplacian matrix aggregating the graph feature information, and $\mathbf{O}$ is initially calculated using $Z$ .

\begin{figure}[t]
\centering
\subfigure[Overall Hypergraph Diagram]{
    \begin{minipage}[t]{0.4\linewidth}
        \begin{center}
        \centerline{\includegraphics[width=\columnwidth]{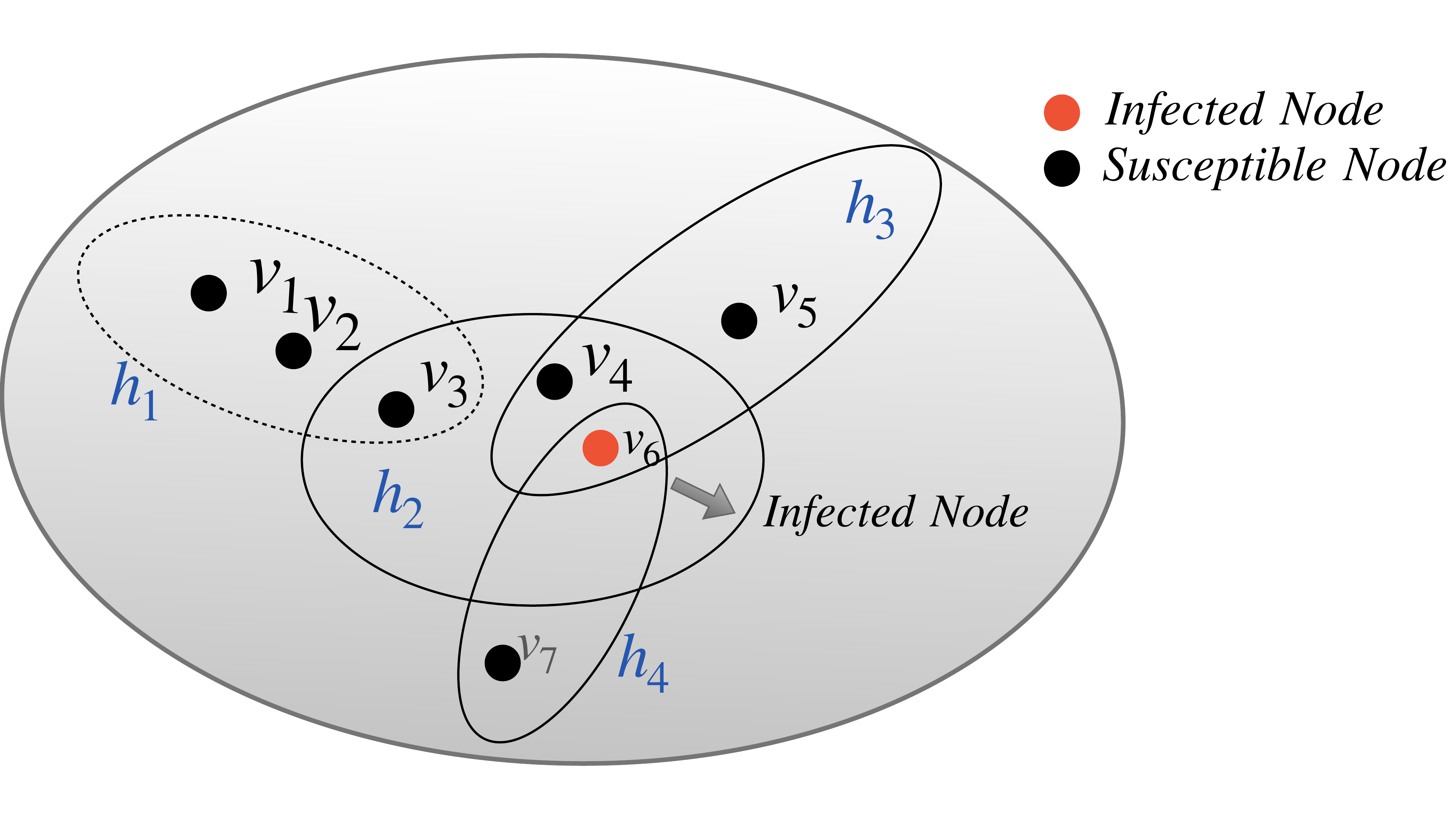}}
        \label{pic1a}
        \end{center}
    \end{minipage}
    }
\subfigure[Spread Process]{
\begin{minipage}[t]{0.4\linewidth}
            \begin{center}
            \centerline{\includegraphics[width=\columnwidth]{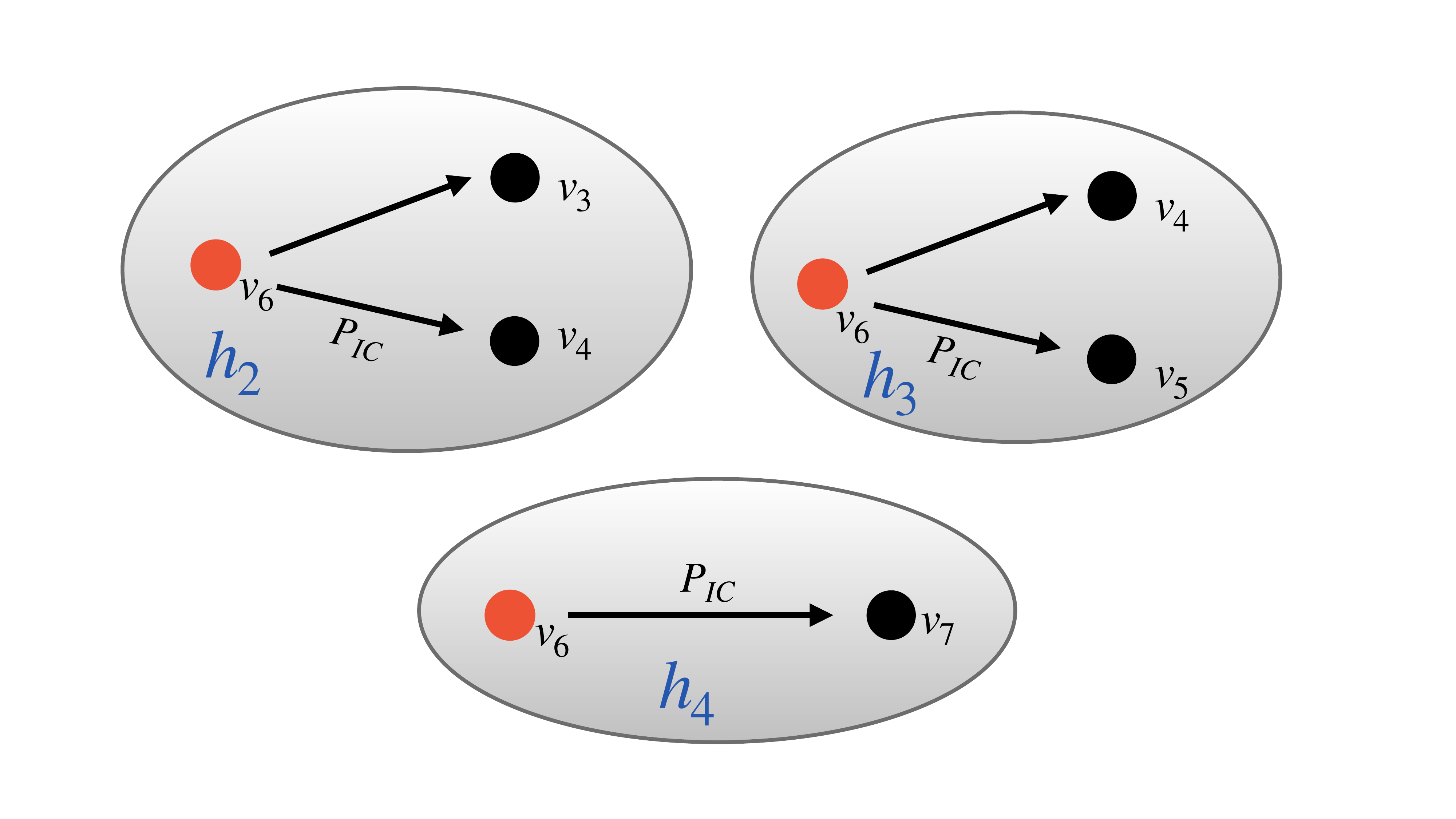}}
            \label{pic1b}
            \end{center}
    \end{minipage}
      
}

\caption{{An illustration of propagation process. With $v_6$ infected, the hyperedges ($h_2, h_3 $ and $h_4$ with solid line) containing $v_6$ are chosen to be the candidates. Nodes ($v_3, v_4, v_5, v_7$) in these hyperedges will potentially convert to the infected state. The spreading probability $P_{IC}$ is the internal parameter in hypergraph-based Independent Cascading (IC) model~\citep{xie2022influence}.}
}
 \label{pic1} 
\end{figure}

\paragraph{More discussions upon IM} For the \emph{first} limitation of IM, the exploration of hypergraph-based IM is urgent to be settled. Since the hypergraph structure is consistent with ample real-world scenarios, especially when different nodes in the graph contain high-level, multivariate relationships, where the traditional graph is hard to model efficiently.~Take an example of the disease propagation problem regarded as IM on hypergraph in Figure \ref{pic1a}. Students are connected through social circles, where each circle can be represented as a hyperedge. Different from ordinary graphs, the influence of node $v_6$ is spread not considering the edge consisting of a pair of nodes, but on more its affiliated hyperedges shown in Figure \ref{pic1b}. Existing hypergraph-based studies are mostly separated into two parts: 1) people are committed to developing heuristic methods but with not enough theoretical support~\citep{antelmi2021social, xie2022influence}. 2) people developed fundamental theoretical guarantees only on a specific form of hypergraph structure~\citep{zhu2018social}. In general, a general hypergraph-based IM with theoretical guarantee and high empirical efficiency still needs to be explored.

For the \emph{second} limitation of IM, the original optimization objective needs to be reconsidered in many cases. The previous objective is directly implied as the sum of node numbers, which stems from empirical or even philosophical determinations and lacks rigorous mathematical arguments. This implication is attributed to the over-simplification of real-life situations--current IM methods tend to overlook the dynamic nature of node influence weights (ITE) within their environments. Traditional IM methods like simulation-based~\citep{kempe2003maximizing,leskovec2007cost} and sketch-based~\citep{borgs2014maximizing, tang2014influence,tang2015influence,wang2016bring} ones focusing on maximizing total numbers (or generalized weighted IM) might fail to pursue such maximum total potential benefit (pursue ``larger varying weighted sum'' instead of ``larger number of nodes''). \citet{wang2016efficient} proposed the new weighted IM problem whereas they rely on non-negative assumptions and  and lacks generalizability to complex scenarios involving hypergraphs and varying node weights. Recently, learning-based IM methods~\citep{chen2023touplegdd,kamarthi2019influence,kumar2022influence,panagopoulos2020multi,ling2023deep} mostly learn potential node representations as a marginal gain of node influence, thus guiding the seed node finding process. Sharing different object functions from ours, many of these existing methods might struggle with limited generalization capabilities and result reliability concerns. Overall, there is a critical need to explore novel objective functions.

\paragraph{More discussions upon extended algorithms} To simplify the discussion, in the additional algorithms we provide, we have omitted the process of dynamically updating ITE based on the surrounding nodes’ state changes during each propagation. CELF-CauIM (Algorithm \ref{alg1: traditional_cauim_mc}) and Monte Carlo-based greedy CauIM (Algorithm \ref{alg_celf}) can be derived naturally from G-CauIM.

\begin{algorithm}[htpb]
\caption{ CELF-{CauIM}} 
\label{alg_celf} 

\begin{algorithmic}[1] 

\REQUIRE Hypergraph $\mathcal{G}(\mathscr{V}, \mathscr{H}, \mathbb{H})$, size of the seed set $K$, causal influence function $\sigma(\cdot)$. 

\STATE Initialization: $S^* = \emptyset$, $MargDic = \emptyset $.
\STATE \textbf{ITE recovery}.

\STATE $MargDic$  stores the marginal gain  $\sigma(\{v\})$ of each node.

\STATE Sort $MargDic$ in decreasing order of value.

\FOR{$|S^*| <K$}
\STATE Move out the node $cur$ with the largest marginal gain in $MargDic$.
\STATE Re-compute marginal gain of of node $cur$ with the current seed set $S^*$: $MargDic[ cur ] = \sigma\left(S^* \cup\{cur\}\right)-\sigma\left(S_{0}\right)$
\STATE Check if previous top node stays on top after sort $MargDic$ again. If true, $S^*=S^*{} \cup\left\{cur\right\}$, and find the second seed; else remove the second largest marginal gain node in $MargDic$, then repeat the last operation.

\ENDFOR
\ENSURE The deterministic seed set $S^*$ with $|S^*| = K$.

\end{algorithmic} 
\end{algorithm}











\begin{algorithm}[htpb]
\caption{ HDD-{CauIM}} 
\label{alg_hdd} 

\begin{algorithmic}[1] 

\REQUIRE $\mathcal{G}(\mathscr{V}, \mathscr{H}, \mathbb{H})$, size of the seed set $K$, causal influence function $\sigma(\cdot)$.

\STATE Initialization: $S^* = \emptyset$, $DegITE = \{\} $.

/* Presented in Algorithm 1 */
\STATE \textbf{ITE recovery}.

\STATE  $DegITE[v]$ stores the sum of ITE of the neighbour nodes of each node $v$, where $N_r(v)$ represents the neighbour nodes of $v$: $DegITE[v] = \sum_{v_r \in N_r(v)} \mathbb{E}\tau_{r}$.

\WHILE{$|S^*| <K$}

\STATE Choose $v_0$ with the max value in $DegITE$ as the seed: $v_0 = argmax_v\{DegITE[v]\}$,  $S^*=S^*{} \cup\left\{v_0\right\}$.
\STATE Calculate sum of ITE for each node $v_r$ in neighbors of $v_0$ as edge value $Edge$: $Edge = \sum_{v_{rq} \in N_r(v_r)} \mathbb{E}\tau_{{rq}}$ .
\STATE Remove the edge influence of the chosen seed node $v_0$: $DegITE[v_r] = DegITE[v_r] - Edge$ .

\ENDWHILE
\ENSURE The deterministic seed set $S^*$ with $|S^*| = K$.

\end{algorithmic} 
\end{algorithm}


For HDD-CauIM~(Algorithm~\ref{alg_hdd}), we replace the selection criterion of searching for nodes with the highest degree with the nodes with the highest sum of the average ITE among the neighboring nodes (line $5-6$). Note that according to the common drawback of heuristic methods, this type of method has no theoretical support and often falls into a common dilemma: nodes with the highest local degree (neighboring ITE) may not necessarily represent seeds that can bring greater overall influence. This point has also been verified in the experiments.

\paragraph{More discussions upon submodularity of hypergraph} {\citet{antelmi2021social, zheng2019non} claimed their hypergraph does not contain submodularity. However, their propagation mechanism is different from traditional IC. {Besides, they considered the form of  directed hypergraph where a hyperedge $(H, t)$  comprises a set of head nodes $H$ is and a single tail node $t$.\iffalse \zhiheng{revised}\fi} Further, \citet{gangal2016hemi} demonstrated the submodularity of a general class of hypergraph. \iffalse \cite{wang2016efficient} proposed the new weighted influence maximization problem. However, their attributes corresponding to each node is a priori assumed to be non-negative, which is different from general ITE (ITE can be positive or negative). \fi Moreover, \citet{erkol2022effective} stated that the submodularity on the temporal network might not be held.}

\paragraph{More discussions upon additional challenges compared to other sum-weighted IMs.} Noteworthy, according to the ITE estimation form, CauIM can be seen as the generalized case of the weighted IM. However, our task is significantly more challenging. Firstly, beginning with the traditional graph, sum-weighted schemes often ensure an approximate optimum guarantee effortlessly, for traditional IM can be extended to weighted IM naturally~(see definition~5~in~\citet{mossel2007submodularity}). However, it does not make sense in our setting since the ITEs for each node would not guarantee to be always non-negative (for instance, some non-compliers, i.e., $\rm{ITE}<0$, exist). Furthermore, the argument on submodularity is more complex in hypergraph (defer to Appendix~\ref{related}). Besides, the ITE estimation does not remain constant between iterations due to the experimental sensitivity (the activation status of surrounding nodes change). In sum, to the best of our knowledge, in the setting of hypergraph-based ITE, the weakened version of the approximate optimum guarantee is an effective supplement to the IM community.

\end{document}